\documentclass[11pt]{article}
\usepackage{graphicx}
\usepackage{amsfonts}
\usepackage{amsmath}
\usepackage{amssymb}
\usepackage{url}
\usepackage{graphicx}
\usepackage{tikz}
\usepackage{fancyhdr}
\usetikzlibrary{calc}
\usetikzlibrary{positioning}
\usetikzlibrary{arrows}

\usepackage{indentfirst}
\usepackage{enumerate}
\usepackage{dsfont,footmisc}
\usepackage{amsthm}
\usepackage{color}
\usepackage{comment}
\usepackage{natbib}
\usepackage{dsfont}
 \usepackage[colorlinks=true,citecolor=blue]{hyperref}

\usepackage[misc]{ifsym}

\tikzset{
	block/.style={rectangle, draw, rounded corners, text centered, text width = 7em, minimum height = 2em},
	line/.style={draw, -latex'}
}
\def\d{\,\mathrm{d}}
\def\laweq{\buildrel \d \over =}

\newcommand{\VaR}{\mathrm{VaR}}

\newcommand{\ES}{\mathrm{ES}}
\newcommand{\E}{\mathbb{E}}
\newcommand{\R}{\mathbb{R}}
\newcommand{\calR}{\mathcal{R}}
\newcommand{\M}{\mathcal{M}}

\newcommand{\p}{\mathbb{P}}

\newcommand{\id}{\mathds{1}}

\newcommand{\X}{\mathcal X}
\renewcommand{\P}{\mathcal P}

\newcommand{\F}{\mathcal{F}}
\newcommand{\Q}{\mathcal Q}

\newcommand{\essinf}{\mathrm{ess\mbox{-}inf}}
\renewcommand{\ge}{\geqslant}
\renewcommand{\le}{\leqslant}
\renewcommand{\geq}{\geqslant}
\renewcommand{\leq}{\leqslant}
\renewcommand{\epsilon}{\varepsilon}
\theoremstyle{plain}
\newtheorem{theorem}{Theorem}

\newtheorem{proposition}{Proposition}

\theoremstyle{definition}
\newtheorem{definition}{Definition} 
\newtheorem{assumption}{Assumption} 
\newtheorem{example}{Example} 
\theoremstyle{remark}
\newtheorem{remark}{Remark} 
\theoremstyle{definition}

\renewcommand{\cite}{\citet} 

\usepackage{setspace}
\newcommand{\com}[1]{\marginpar{{\begin{minipage}{0.18\textwidth}{\setstretch{1.1} \begin{flushleft} \footnotesize \color{red}{#1} \end{flushleft} }\end{minipage}}}}

\begin{document}
\title{A Framework for Measures of Risk under Uncertainty}

\author{Tolulope Fadina\thanks{School of Mathematics, Statistics and Actuarial Science, University of Essex, UK. \Letter~\texttt{t.fadina@essex.ac.uk}.}
	\and Yang Liu\thanks{Corresponding author. School of Science and Engineering, The Chinese University of Hong Kong, Shenzhen, China.  \Letter~\texttt{yangliu16@cuhk.edu.cn}.} 
	\and Ruodu Wang\thanks{Department of Statistics and Actuarial Science, University of Waterloo, Canada. \Letter~\texttt{wang@uwaterloo.ca}.}}

\date{}%\today} 
\maketitle
%School of Mathematics, Statistics and Actuarial Science
\begin{abstract}

	A risk analyst assesses potential financial losses based on multiple sources of information. Often, the assessment does not only depend on the specification of the loss random variable but also various economic scenarios. Motivated by this observation, we design a unified axiomatic framework for risk evaluation principles which quantifies jointly a loss random variable and a set of plausible probabilities. We call such an evaluation principle a generalized risk measure. We present a series of relevant theoretical results. The worst-case, coherent, and robust generalized risk measures are  characterized via different sets of intuitive axioms. We establish the equivalence between a few natural forms of law invariance in our framework, and the technical subtlety therein reveals a sharp contrast between our framework and the traditional one. Moreover, coherence and strong law invariance are derived from a combination of other conditions, which provides additional support for coherent risk measures such as Expected Shortfall over Value-at-Risk, a relevant issue for risk management practice. 
	
	%	A risk analyst conducts financial assessment not only based on the loss random variable, but also considers various economic scenarios. In this regard, we design a unified axiomatic framework that uses both the loss random variable and a set of underlying probabilities, as the input of a risk measure which is hence called the generalized risk measure. As the most practical preference, the worst-case generalized risk measure is characterized by axioms and illustrated with concrete examples. In our framework, law invariance is   crucial in both theoretical analysis and practical application. We reveal the relationship between a few natural forms of law invariance, under which we further express the worst-case generalized risk measures in form of functions defined on the distributions. Moreover, we show that strong law invariance of coherent generalized risk measures can be obtained from loss law invariance and several other properties. This result provides additional evidence for traditional coherent risk measures. Finally, robust generalized risk measures are studied with their characterizing properties. 
	
	~
	
	\noindent \begin{bfseries}Keywords\end{bfseries}:  Risk management; model uncertainty; regulatory capital; variational preferences; law invariance; decision theory.
	%Basel IV; Law invariance; Robustness; Decision theory.

\end{abstract}
~

\newpage

\newpage

\section{Introduction} \label{sec:intro}
 
Risk measures are widely used in both financial regulation and economic decisions.  Since the seminal work of \cite{ADEH99}, risk measures are commonly defined as  functionals on a space of random variables, or, with the assumption of law invariance, on the set of their distribution functions.  The most popular risk measures are the Value-at-Risk (VaR) and the Expected Shortfall (ES); see \cite{ADEH99} and \cite{FS16} for the classic theory of risk measures,  %\cite{WZ21} for more recent advances on the axiomatization of ES, 
and documents from the Basel Committee on Banking Supervision (BCBS), e.g., \cite{BASEL19}, for regulatory practice in banking.

In this paper, we propose a novel framework for measures of risk under uncertainty. 
Let us first explain our motivation. We take market risk as our primary example, although our discussions naturally apply to many other types of risks. A portfolio is associated with a future loss random variable $X$ representing the portfolio risk.
The loss $X$ has two important practical aspects: the \emph{specification} and the \emph{modeling}. 

\begin{enumerate}
	\item 
	The specification refers to how $X$ is defined in terms of the underlying risk factors (e.g., asset prices, exchange rates, credit scores, volatilities, etc). More precisely, $X$ is the financial loss (or gain) from holding assets, derivatives, or other investments in the portfolio. 
	Mathematically, the specification of $X$ is represented by the function $X:\Omega\to \R$, which maps each state of the future financial world (each element of the sample space $\Omega$) into a realized loss.
	\item The modeling refers to the statistical assessment of the likelihood and the severity of the loss $X$.
	The modeling of $X$ is usually summarized by a distribution, or a collection of distributions in case of model uncertainty, under some estimated or hypothetical (e.g., in stress testing) probability measures $\p\in \P$, where $\P$ is the set of probability measures on the sample space $\Omega$.
\end{enumerate}

In the classic framework of \cite{ADEH99}, a  risk measure $\rho$ is defined on a set $\X$ of random variables, and the risk value $\rho(X)$ is thereby determined by the specification of $X$. 
The modeling of $X$ is, however, implicit in this setting: through a given probability $\p $, assuming available to the end-user,
the distribution of $X$ under $\p$ is determined by its specification. 

There is a visible gap in the classic setting $\rho:\X\to\R$. In practice, neither $X$ nor $\p$ is generally fixed.
A change in $X$ means adjusting positions via trading financial securities. 
A change in $\p$ means an update of the modeling, estimation, and calibration of the random world. 
In financial practice, both $X$ and $\p$ evolve on a daily basis for a trading desk; yet they are modified daily for very different reasons. 

For another concrete example, suppose that a regulator specifies a risk measure (e.g., ES at level $0.975$ as in \cite{BASEL19}), and two firms assess the risk of the same portfolio separately. Due to different modeling and data processing techniques used by the two firms, their reported ES values are typically not the same. However, the loss random variable $X$ from the portfolio is the same for both firms. Therefore, the risk measure should not be only determined by the specification of $X$, but also the modeling information. In practice, modeling is always subject to uncertainty (called ambiguity in   decision theory). Even in the simple estimation of a parametric model, the plausible models are not unique; see \cite{GS89} for a classic treatment of ambiguity.

Motivated by the above observations, we propose a new framework of risk measures taking into account both the specification and the modeling of random losses.
We choose a set of probability measures instead of a single probability measure as the input for the modeling component. Formally, we introduce the \emph{generalized risk measure} $\Psi:\X\times 2^\P \to [-\infty,\infty]$, which has two input arguments: a random variable $X\in \X$ representing the specification of the loss,  and a set of probability measures $\mathcal Q\subseteq \P$ representing the modeling of the random world; each probability measure in $\P$ is called a \emph{scenario}.  Our framework includes the traditional law-invariant risk measures as a special case when $\mathcal Q$ is a pre-specified singleton. When $\mathcal Q$ is fixed but not a singleton, our generalized risk measures include the scenario-based risk measures of \cite{WZ18}. The framework also incorporates other complicated decision criteria addressing model uncertainty in the literature, which will be discussed later. 
% are mapping from $\X\times 2^\P$ to $\R$, where $\X$ is a space of random variables, and $2^\P$ is the set of subsets of $\P$, probability measures.   

We take the perspective of a regulator, who designs a regulatory capital assessment scheme that will be complied with by financial institutions. Financial institutions (or their trading desks) can choose their portfolio positions with losses $X\in \X$,
and, subject to passing regulatory backtests for statistical prudence, they can also choose their internal models $\mathcal Q\subseteq \P$. The generalized risk measure is crucial to the design of the capital assessment procedure, because it acts on portfolios and models from financial institutions, and computes regulatory capital requirements. Therefore, our theoretical framework closely resembles the regulatory practice in the Fundamental Review of the Trading Book (FRTB) of \cite{BASEL19}; see \cite{WZ18} for discussions on the risk assessment practice of FRTB,
and \cite{CF17} for  model and scenario aggregation methods in solvency assessment. 
 It is important to note that the input scenario set $\mathcal Q$ does not necessarily contain the decision maker's subjective probability governing the random world, because most models are simplifications or approximations, as argued by \cite{CHMM21}. 

Figure \ref{motivating-example} illustrates  a stylized risk assessment procedure, reflecting many of the above considerations. There are four roles: regulator (external), risk analyst (internal), portfolio manager (internal), and model risk manager (internal). In Figure \ref{motivating-example}, except for the regulator's actions, the other actions are changing dynamically on a daily (or similar) basis, making it clear that one should take both $\mathcal Q$ and $X$  as inputs and allow them to vary in a unified framework.
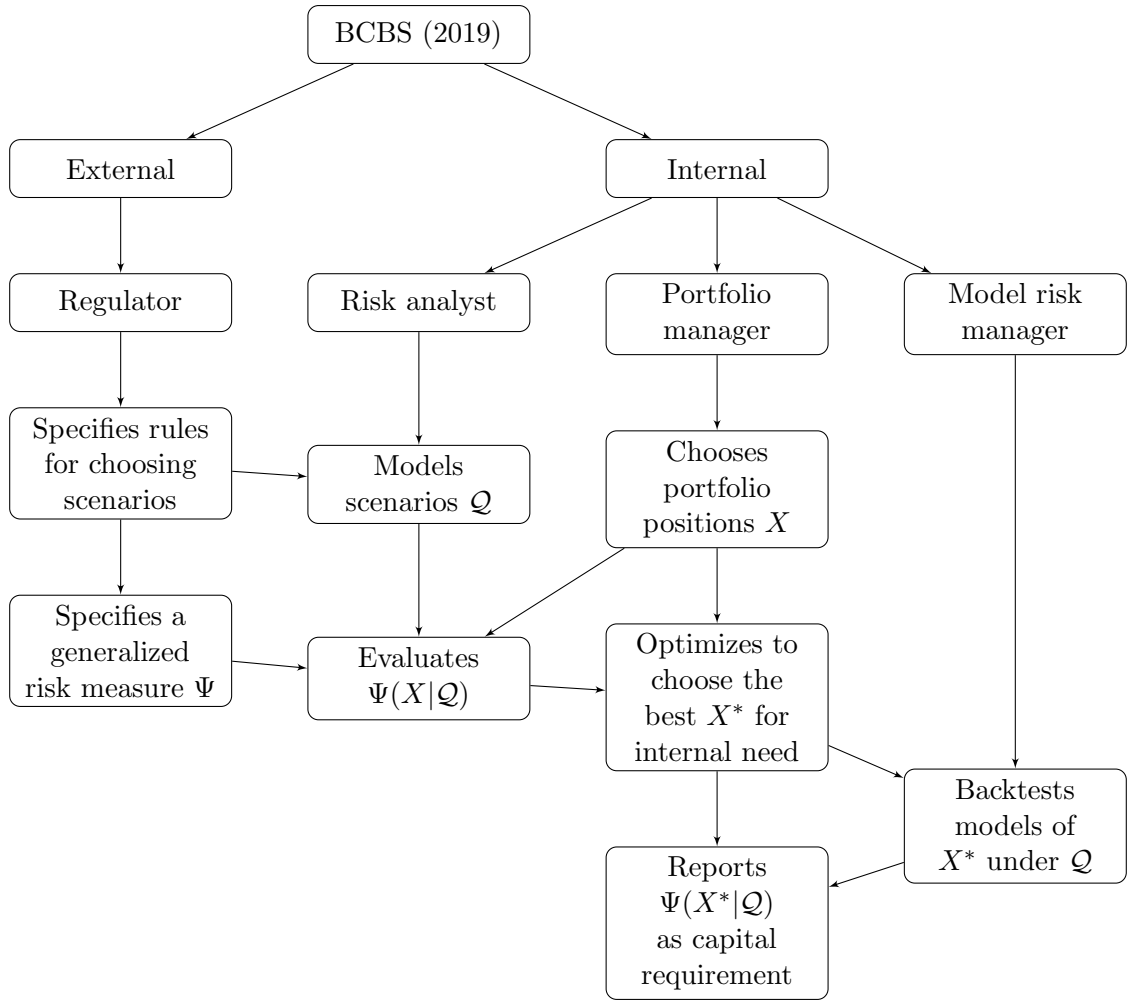
\begin{figure}		
	\begin{center}
		\begin{tikzpicture}[every node/.style={block}]
			% Place nodes
			\node (topic) {BCBS (2019)};
			\node[below left=1cm and 1cm of topic](case one) {External};
			\node[below right = 1cm and 1cm of topic](case three) {Internal};
			\node[below = 1cm of case one](regulator){Regulator};
			\node[below = 1cm of regulator](role regulator){Specifies rules for choosing scenarios};
			\node[below = 1cm of role regulator](role regulator2){Specifies a generalized risk measure $\Psi$};% (such as worst-case ES)};
			\node[below left = 1cm and 1cm of case three](risk) {Risk analyst};
			\node[below = 1.5cm of risk](estimate) {Models scenarios $\mathcal Q$};
			\node[below = 1.5cm of estimate](evaluate) {Evaluates $\Psi(X\vert \mathcal Q)$};
			\node[below = 1cm of case three](more){Portfolio manager};
			\node[below = 1cm of more](choose){Chooses portfolio positions $X$};
			\node[below = 1cm of choose](optimize){Optimizes to choose the best $X^*$ for internal need};
			\node[below = 1cm of optimize](report){Reports $\Psi(X^*\vert \mathcal Q)$ as capital requirement};
			\node[below right = 1cm and 1cm of case three](model) {Model risk manager};
			\node[below = 5.5cm of model](backtest) {Backtests models of $X^*$ under $\Q$};

			\foreach \x/\y in {
				topic/case one,
				topic/case three,
				case one/regulator,
				regulator/role regulator,
				role regulator/estimate,
				role regulator/role regulator2,
				case three/more,
				case three/risk,
				case three/model,
				more/choose,
				risk/estimate,
				estimate/evaluate,
				role regulator2/evaluate,
				choose/evaluate,
				evaluate/optimize,
				choose/optimize,
				optimize/report,
				optimize/backtest,
				model/backtest,
				backtest/report}
			\draw [line] (\x) -- (\y);
		\end{tikzpicture}
	\end{center}
	\caption{A stylized procedure for risk assessment practice}
	\label{motivating-example}
\end{figure}

\subsection{Contribution and structure of the paper}
As explained above, in the literature on the axiomatic theory of risk measures, one often first  designs axioms to identify desirable risk measures without model uncertainty, and then puts model uncertainty into the model as an exogenous object. This approach, although easy to apply, is unsatisfactory from a decision-theoretic point of view, as it does not identify desirable axioms for risk measures when model uncertainty is taken as input. One of our main contributions is to provide an axiomatic framework of generalized risk measures which allows us to consider properties on both the model uncertainty and the random losses, thus addressing this practical issue for the first time. The rigorous mathematical formulation of generalized risk measures is laid out in Section \ref{sec:framework}.

Since generalized risk measures are defined as a mapping from $\X\times 2^\P$ to the (extended) real line, the mathematical structure is much more complicated than traditional risk measures.  
We establish several relevant theoretical results related to this new framework. 
In Section \ref{sec:worst}, we obtain an axiomatic characterization of  worst-case generalized risk measures via some simple properties (Theorem \ref{th:worstcase}). The worst-case generalized risk measures are the most practical and they appear extensively in the literature on risk measure and optimization.

%As the most practical preference, the worst-case generalized risk measure is investigated. It is characterized by axioms and illustrated with concrete examples; see Section \ref{sec:worst}. 
%Again, for these results, the set of probabilities is regarded as an input. However, the literature on robust optimization typically considers a set of a random variable, and this set is considered as an objective constraint.

%\com{Please change "law-invariance" to "law invariance" throughout. Do keep "law-invariant" (adjective; actually you should use both law invariant and law-invariant if you know the grammatical difference). Same for ``ambiguity sensitivity" and ``ambiguity-sensitive" (actually you should use both ambiguity sensitive and ambiguity-sensitive if you know the grammatical difference).}
Law invariance is a crucial property that connects loss random variables to statistical models. In the traditional framework, law-invariant risk measures can be equivalently expressed as functionals on a space of distributions; this is no longer true in our generalized framework. We provide three different forms of law invariance which reflect different considerations: strong law invariance, loss law invariance, and scenario law invariance; see Section \ref{sec:law invariance} for details.
In general, the three notions of law invariance are not equivalent and they reflect very different modeling considerations. Indeed, if strong law invariance is assumed, our framework can be converted to the traditional setting without many mathematical difficulties. %; we include some characterization results in Appendix \ref{sec:app} assuming strong law invariance. 
However, in practice, strong law invariance may not be desirable, and technical complications arise when it has to be weakened. 
In Section \ref{sec:law invariance}, we show the equivalence between strong law invariance and a combination of the two weaker versions of law invariance under mild conditions (Theorem \ref{thm:law-inv}). Moreover, we express the worst-case generalized risk measures with various kinds of law invariance as functions defined on the distributions (Proposition \ref{th:law invariance}). Therefore, from traditional law-invariant risk measures defined on distributions, we can easily construct generalized risk measures satisfying certain desirable properties. 

In Section \ref{sec:connect}, we focus on coherent generalized risk measures, an analog to the coherent traditional risk measures of \cite{ADEH99}, and characterize the simplest form (expectation-type) in Theorem \ref{th:coh}. Moreover, we propose the notion of ambiguity sensitivity and establish an equivalence between strong law invariance and a combination of a weaker law invariance and ambiguity sensitivity (Theorem \ref{th:new2}). In addition, the combination, with a few simple properties, implies coherence, which supports coherent risk measures in the traditional framework from a completely novel perspective.

In Section \ref{sec:connection}, we discuss some connections of our framework to decision theory. Particularly, we characterize the multi-prior expected utility of \cite{GS89} with several properties (Proposition \ref{th:preference}) and obtain an axiomatic characterization for the robust generalized risk measure (Proposition \ref{th:robust}), closely related to the variational preferences of \cite{MMR06}. %Ideally, with the assumption of strong law invariance, our setting can be reduced to traditional framework without many technical difficulties.
Section \ref{sec:conclusion} contains further discussions and remarks. Proofs of all theorems and propositions are put in the appendix.

\subsection{Connections to other frameworks in the literature}
Our framework is in sharp contrast to the existing ones in the literature on risk management.  We have already discussed the difference between our framework and the classic frameworks of risk measures (see \cite{ADEH99} and \cite{FS16}) or preferences (see \cite{W10} for a comprehensive treatment), which are all defined on $\X$. %We only discuss the ones which take model uncertainty into account.  
The setting of scenario-based risk measures of \cite{WZ18} is  also motivated by the regulatory framework of \cite{BASEL19} and aims to understand uncertainty in risk measures, but it is mathematically quite different.
Scenario-based risk measures are mappings on $\X$  determined by   the distributions of the random losses   under a collection of pre-specified scenarios. Since the scenarios are fixed, the key question of how a risk measure reacts when scenarios change is left unaddressed. As such, the mathematical results in this paper have no overlap with \cite{WZ18}.

Model uncertainty is an important topic in economic decision theory. In the classic setting of \cite{AA63}, a risk (called a lottery) is represented by a collection of possible distributions, whereas, in our framework,  the input consists of a random variable and  a collection of  probability measures, which interact with each other. %Although the two settings address similar practical considerations, their mathematical structures are very different. %, as discussed by \cite{WZ18}. For instance, monotonicity on random variables cannot be easily translated to monotonicity on distribution functions; indeed, such a translation is harmless if there is no model uncertainty (through stochastic dominance), but it is impossible if there are multiple scenarios.  Monotonicity is crucial for risk measures, and this difference leads to drastically different axiomatic representation results.
There are many recent developments in this stream of literature, which focus on the characterization of preferences under uncertainty via some axioms. For a non-exclusive list, we mention  the multi-prior expected utility of \cite{GS89}, the multiplier preferences of \cite{HS01}, the smooth ambiguity preference of \cite{KMM05}, the variational preference of \cite{MMR06},  and the model misspecification preference of  \cite{CHMM21}. They can be formulated as  examples of our framework, which will be illustrated in Example \ref{ex:literature}.

 Some conceptual frameworks in decision theory reflect similar considerations towards risk and uncertainty as ours. In particular, 
\cite{CHMM21} studied preferences under model misspecification, and their set of structured models corresponds to our set of scenarios $\mathcal Q$. 
An earlier  work closely related to our framework is  \cite{GHTV08}, where the authors studied preferences defined on the outcome mapping (an act) and the set of possible probabilities; thus conceptual similarity is clear. 
%Such a framework reflects similar considerations as ours.  
Nevertheless, since the main context  of our work is financial risk assessment instead of decision making, the axioms and properties considered in this paper, as well as  technical results and their implications, are completely different from \cite{GHTV08} and \cite{CHMM21}. 

In the operations research literature,   \cite{DKW19} recently investigated a model for decision making with and without uncertainty, and analyzed the conditions under which random decisions are strictly better than deterministic ones. 
Model uncertainty also widely appears in   robust optimization; see \cite{EOO03}, \cite{ZF09} and \cite{ZKR13} for optimizing risk measures under uncertainty. In the above literature, model uncertainty is generally  pre-specified and regarded as an objective fact, whereas we study the properties of risk measures taking model uncertainty as an input argument that can vary over all possible choices.

\section{A framework for measures of risk and uncertainty}
\label{sec:framework}

\subsection{Notation} 

We begin by stating some notation which will be used throughout. 
Let $(\Omega, \mathcal F)$ be a measurable space and let $\mathcal P$ be the class of atomless probability measures   defined on the measurable space.\footnote{A probability measure $P$ on $(\Omega, \mathcal F)$ is atomless if there exists a uniform random variable on $(\Omega, \mathcal F,P)$.}
The set of all subsets of $\mathcal P$ is denoted by $2^\mathcal P$.
Let $\mathcal X $ be the space of 
bounded
%\trd{bounded?}
random variables,
and  $\M$ be the set of compactly supported distributions on $\R$.
% equipped with the $ L^{\infty} $-norm $\left\\vert \cdot \right\\vert _{ \infty}$, that is $$\left\\vert  X \right\\vert _{ \infty} = \inf \{  M\geq 0 :  \lvert  X\rvert  \leq M  \text{ P-a.s for every } P\in \mathcal P  \}.$$Under the norm $\left\\vert  \cdot \right\\vert _{ \infty}$,  $L^{\infty}$ is a Banach space, \cite[Proposition 15]{Denis}. 
For $X,Y\in \X$ and $P, Q \in \mathcal P,$ we write
$X\vert _P \laweq Y  \vert  _Q$
if the distribution of $X$ under $P$ is identical to that of $Y$ under $Q$. 
We denote by $F_{X\vert P} \in \M$ the distribution of $X$ under $P$. 
For an increasing set function $\nu : \mathcal F\to \R$ with $\nu(\varnothing)=0$, 
a Choquet integral (e.g., Definition 4.76 of \cite{FS16}) with respect to $\nu$ is defined as  
$$
\int X \d \nu = \int_{-\infty}^{0} (\nu (X \geq x) - \nu ( \Omega)) \d x +  \int_{0}^{\infty} \nu (X \geq x)  \d x ,~~~~ X \in \mathcal X.
$$  
We note in the following example that the same random variable $X$ can be continuously distributed under one atomless probability measure $P$ and discretely distributed under another atomless probability measure $Q$. Working with atomless probability measures allows   us to study  continuously distributed as well as discrete  random variables in a unified framework. 
\begin{example}  Consider $(\Omega,\mathcal F)=([0,1]^2,\mathcal B([0,1]^2))$ where $\mathcal B$ is the Borel $\sigma$-algebra. 
	Let $P = \lambda  \times \lambda$ and $Q=\delta_1 \times \lambda$, where $\lambda$ is the Lebesgue measure on $[0,1]$ and $\delta_1$ is the point mass at $1$.  Note that both $P,Q$ are atomless probability measures. Let $X(s,t)=s$ for $(s,t)\in [0,1]^2$. 
	The distribution of $X$ under $P$ is  the uniform distribution on $[0,1]$ and the distribution of $X$ under $Q$ is the point mass at $1$.  
\end{example}

%For $X \in \mathcal X$, $X \sim F$ means that $F$ is the cumulative distribution function of $X$ under $P$ for every $P  \in  \mathcal P$. The left-continuous generalized inverse of $F$ (left-quantile) is denoted by 
%$$F^{-1}(t) = \inf \{x \in \mathbb R \lvert F( x) \geq t  \},  \quad t  \in ( 0,1],  \text{ and } F^{-1}(0) = \sup \{x \in \mathbb R \lvert F( x) = 0  \}$$
%whereas its right-continuous generalized inverse of $F$ (right-quantile) is
%$$F^{-1+}(t) = \sup \{x \in \mathbb R \lvert F( x) \leq t  \},  \quad t  \in ( 0,1],  \text{ and } F^{-1+}(1) = F^{-1}(1)$$
%is the right (upper) quantile function of $X$. 

%Let $ h : [0,1] \rightarrow [0,1] $ be an increasing function such that $h(0) = 0$ and $ h(1) = 1$. The set function 
%$$v (A) : = h(P(A)),  \quad A  \in  \mathcal F, $$
%is called the \emph{distortion} of the probability measure $P$ with respect to the distortion function $h$.  A set $v :  \mathcal F \rightarrow [0,1] $ is called monotone of $v (A)  \leq h (B)$ for $A  \subset B$ and normalized if $v(  \emptyset) = 0$ and $v(  \Omega) = 1$ . 
%A monotone set function is called submodular or strongly subadditive if 
%$$v(A \cup B ) + v(A \cap B ) \leq v(A) +  v(B).$$
%Clearly, any distortion $v$ is normalized and monotone. 
% 

\subsection{A new and generalized framework for risk measures}

Traditionally, risk measures in the sense of \cite{ADEH99} and \cite{FS16} are mappings from $\X$ to $\R$. We will call them  \emph{traditional} risk measures. 
Note that a traditional risk measure does not require the specification of a probability measure unless we additionally assume law-invariance; 
  this will be further discussed in Section \ref{sec:law invariance}.

In the new framework that we work with in this paper, the input of a risk measure 
is a combination of the loss $X$ and a set of possible probability measures $\mathcal Q$ that represents the best knowledge of the underlying random nature. 
To distinguish from the traditional  setting, we shall refer to these functionals as \emph{generalized} risk measures.  

\begin{definition}\label{def:1}
	A  \emph{generalized risk measure}\footnote{We use the notation $\Psi(X\vert \mathcal Q)$ instead of $\Psi(X,\mathcal Q)$ to emphasize the different roles of $X\in \X$ and $\mathcal Q\in 2^{\P}$.}  is a mapping $\Psi:\X\times 2^\P \to [-\infty,\infty]$. The generalized risk measure $\Psi$ is called \emph{standard} if ${\Psi(s\vert \Q) = s}$ for all ${s\in \R}$ and ${\Q\subseteq \mathcal P}$.
	A \emph{single-scenario risk measure} is a mapping $\Psi:\X \times \P\to [-\infty,\infty]$. 
\end{definition}

Clearly, a single-scenario risk measure is precisely a generalized risk measure with its second argument confined to singletons of scenarios.  For a singleton $\{P\}\subseteq \mathcal P$, we use the simpler notation ${\Psi(X\vert P)=\Psi(X\vert \{P\})}$. %{\color{red}
	 For any fixed $P$, the mapping $X\mapsto \Psi(X\vert P)$is a risk measure in the traditional sense.%. \sout{For a fixed $\mathcal Q$,  $\Psi(\cdot\vert \mathcal Q)$ is a risk measure in the traditional sense}. 

The requirement of standardization reflects the consideration that, for any fixed constant $s$, %{\color{red}
	$\Psi(s\vert \mathcal Q)$ does not depend on the input scenarios $ \mathcal Q$. %}\sout{there is no uncertainty associated with it, and $\Psi(s\vert \mathcal Q)$ should not depend on the input scenarios $ \mathcal Q$.}
The range of $\Psi$ is chosen as $[-\infty,\infty]$ in our general framework to allow for the greatest generality. In practical applications, one may restrict the range to be $\R$  or $(-\infty,\infty]$.

%In the classic literature (e.g., \cite{ADEH99} and \cite{FS16}), risk measures are mappings from $\X$ to $\R$, which we shall call  \emph{traditional risk measures}. 
Generalized risk measures are much more complicated as a mathematical object than traditional risk measures, since their input includes both a random loss $X$ and a set of probability measures $\mathcal Q$. 
Below we collect some basic properties to consider for a generalized risk measure $\Psi$.  
\begin{enumerate}[(A1)]
	\item[(A1)]\label{item:UA} Uncertainty aversion:   
	$\Psi(X\vert  \mathcal Q)\le \Psi(X\vert \mathcal R) $ for all $X\in \X$ and $\mathcal Q\subseteq \mathcal R\subseteq \mathcal P$.
	%	\item[(A2)] $\Psi(X\vert  \mathcal Q) =  \Psi  (X\vert  \{P\} )$ if $F$ is the distribution of $X$ under $Q$ for all $Q \in \mathcal Q $ and $P \in \mathcal Q $. 
	\item[(A2)]\label{item:UM} Scenario monotonicity: $\Psi(X\vert \Q) \leq \Psi(Y\vert \Q)$ if
	$\Psi(X\vert  P ) \leq \Psi(Y\vert  P )$ for all $ P \in \Q$.
	\item[(A3)]\label{item:UB} Scenario upper bound: $\Psi(X \vert \Q) \leq \sup_{P\in \mathcal Q} \Psi  (X\vert P)$ for all $X \in \X$ and $\mathcal Q \subseteq \P$.
	%	\item[(A3)] Normalization: $\Psi(s\vert \Q) = s$ for all $s\in \R$ and $\Q\subseteq \mathcal P$.
\end{enumerate}
Property (A1) means that the evaluation of the risk weakly increases if model uncertainty increases, and this reflects an aversion to model uncertainty. 
Property (A2) means that, if under each possible scenario, $X$ is evaluated to be less risky than $Y$, then the overall evaluation of the risk of $X$ should not be more than $Y$. 
Property (A3) means that the overall evaluation of $X$ is not more extreme than $X$ evaluated under the worst-case scenario.

(A2) and (A3) are quite natural and they are satisfied by most examples of generalized risk measures in their various disguises in the risk management and decision theory literature; we will discuss some of them later. 

Property (A1) is more specialized, as it will lead to worst-case risk evaluation or decision making (Theorem \ref{th:worstcase} below) axiomatized in decision theory by \cite{GS89}. This property is not satisfied in models where uncertainty is aggregated in some form of averaging, such as taking a weighted average of risk evaluates such as the average ES of \cite{WZ18} or the smooth ambiguity model of \cite{KMM05}.  Indeed, if a new scenario $P$ is added to an existing collection of scenarios $\mathcal Q$, and a random loss $X$ is considered safe under $P$, then it may be desirable in risk management practice to reduce the assessment of riskiness of $X$  by including $P$, 
that is, 
 $\Psi(X \vert\{\mathcal Q,P\})< \Psi (X\vert \mathcal Q)$, violating (A1). 
 
 In decision theory, after a proper translation between two frameworks, the preferential version of (A2) appears in \cite{CHMM21} as \emph{$Q$-separability}, and (A1) is genuinely weaker than \emph{monotonicity in model ambiguity} of \cite{CHMM21} on preferences.
 For a fixed set of scenarios $\mathcal Q$,  (A1) and (A2) are, respectively,  similar to  \emph{ambiguity aversion} and \emph{ambiguity monotonicity}  in \cite{DKW19}, which are   formulated for distributions rather than random variables.

\subsection{Examples: VaR and ES}\label{sec:23}

We first give a few examples in this section, and more will be discussed later. 
The two popular traditional risk measures in banking and insurance are the Value-at-Risk (VaR) and the Expected Shortfall (ES); see \cite{EPRWB14} for a review. Both risk measures in the classic formulation are defined with a fixed scenario $P\in \P$,
and allowing $P$ to vary we can treat them as single-scenario risk measures in Definition \ref{def:1}. 
For a level $\alpha \in (0,1]$, the VaR  under $P$  is defined as 
$$\VaR _{\alpha} (X\vert P) = \inf \{x \in \mathbb R :   P(X\le  x) \geq \alpha  \} ,~~~~X\in \X,$$ 
and the ES under $P$ is defined as 
$$\ES_{\alpha}(X\vert P) =  \frac{1}{1-\alpha}  \int_{\alpha}^{1} \VaR_\beta (X\vert P) \d \beta,  ~~~~ X\in \mathcal X.$$   

We first show some properties of VaR and ES in our setting. 
These properties follow from existing properties of VaR and ES with a fixed $P$, 
but the concavity or convexity with respect to scenarios  is not formally studied in the literature since our framework is new.\footnote{The statement that $(X,P)\mapsto \VaR_\alpha(X\vert P)$  is neither convex nor concave in $X$ or $P$ may fail if $P$ is an atomic probability measure. 
For instance, if  $P$ is a discrete measure with probability mass $1/n$ on $n$ points,  then $\VaR_\alpha(\cdot | P)=\ES_\alpha(\cdot | P)$ for $\alpha>1-1/n$, making the statement false.
Recall that we work with atomless probability measures throughout. 
}
\begin{proposition}\label{prop:ES}
Fix $\alpha \in (0,1)$. 
The single-scenario risk measure $(X,P)\mapsto \ES_\alpha(X\vert P)$ is convex in $X$ and concave in $P$, whereas  $(X,P)\mapsto \VaR_\alpha(X\vert P)$  is neither convex nor concave in $X$ or $P$.
\end{proposition}

Building on the single-scenario VaR and ES, we can define generalized risk measures such as worst-case VaR and worst-case ES, via
		\begin{equation*}
		 \overline{\mathrm{VaR}}_\alpha(X\vert \mathcal Q) =  \sup_{P\in \mathcal Q} \VaR_\alpha (X\vert P)~~~\mbox{and}~~~\overline{\mathrm{ES}}_\alpha(X\vert \mathcal Q) =  \sup_{P \in \Q} \ES_\alpha(X\vert P),~~~~(X,\mathcal Q)\in \X\times 2^{\mathcal P}.
		\end{equation*} 
 We refer to \cite{EOO03} for optimization of the worst-case VaR,  \cite{ZF09} for optimization of the worst-case ES, and \cite{WZ18} for their theoretical properties.
The worst-case VaR and the worst-case ES are both standard, and they satisfy (A1), (A2) and (A3). 

For a given fixed $\mathcal Q\subseteq \mathcal P$, 
several other examples of ES and VaR with aggregated scenarios, such as averages  (with respect to a pre-specified measure over $\mathcal Q$) 
 and inf-convolutions (for a finite $\mathcal Q$), are also considered by \cite{WZ18} and \cite{CCMTW21}. 
 For instance, we can define an average ES by 
\begin{equation}
\label{eq:avgES}  (X,\mathcal Q)\mapsto \int_{\mathcal {Q}} \ES_\alpha(X\vert Q) \d \mu_{\mathcal Q}(Q),
\end{equation}
where $\mu_{\mathcal Q}$ is a measure over $\mathcal Q$ for each $\mathcal Q\subseteq \mathcal P$. The average ES in \eqref{eq:avgES} is standard and it satisfies (A2) and (A3); it does not satisfy (A1) in general.
 We remark that, although sharing many common forms and examples, our framework is fundamentally different from the existing ones in the literature, as it is crucial for a generalized risk measure to use $\mathcal Q$ as an input variable instead of a pre-specified collection.

\section{Worst-case generalized risk measures}
\label{sec:worst}

In this section, we present our first theoretical result,  a characterization of generalized risk measures satisfying (A1) as the supremum of risk measures in the traditional sense. This result allows us to apply many results on traditional risk measures to generalized risk measures. 

%\subsection{Worst-case uncertainty-risk measure}
\begin{theorem}\label{th:worstcase}
	For a generalized risk measure $\Psi: \X \times 2^\P \rightarrow \R$: %\com{for consistency, should change to (i) and (ii).}
	\begin{enumerate}[(i)] \item Suppose that $\Psi$ is standard. $\Psi$ satisfies (A1)-(A2) if and only if
	it admits a representation \begin{equation}\label{eq:def_unc}
		\Psi(X\vert \mathcal Q) = \sup_{P\in \mathcal Q} \Psi  (X\vert P),~~~~(X, \mathcal Q)\in \X\times 2^\P.
	\end{equation} 
	\item $\Psi$ satisfies (A1) and (A3) if and only if
	it admits a representation \eqref{eq:def_unc}.
	\end{enumerate}
\end{theorem}

Using Theorem \ref{th:worstcase}, we can  pin down the forms of possible generalized risk measures with properties on the simpler object $\Psi(X\vert P)$ for $X\in \X$ and $P\in \mathcal P$.  Theorem \ref{th:worstcase} is a general functional form of the specific preferential characterization treated in Theorem 2 of \cite{CHMM21}. 
\begin{definition}
	For a given generalized risk measure $\Psi$, the single-scenario risk measure $(X,P)\mapsto \Psi(X\vert P)$, $X\in \X$, $P\in \mathcal P$, is called the \emph{core} of  $\Psi$.   
\end{definition}
By Theorem \ref{th:worstcase}, the cores  one-to-one correspond  to standard generalized risk measures that satisfy (A1)-(A2) via \eqref{eq:def_unc}. Note that in general,  the core of $\Psi$ does not determine  $\Psi$ on $\X\times 2^\P$, if the conditions in Theorem \ref{th:worstcase} are not satisfied. 

In case \eqref{eq:def_unc} holds, we say that the core $\Psi$ on $\X\times \mathcal P$ induces the generalized risk measure $\Psi$ on $\X\times 2^{\mathcal P}$. Many results in this paper are stated for cores instead of the generalized risk measure.  Nevertheless, when we speak of cores, we do not need to assume the worst-case form \eqref{eq:def_unc} or any of (A1)-(A3).

Some simple examples for the worst-case generalized risk measures are collected below, and they appear in  forms similar to those in the classic theory of risk measures. 
\begin{example}\label{ex:worst}
	\begin{enumerate}[(i)]
		\item  The expectation core 
		$$\Psi(X\vert P)= \E ^P [X],~~~~(X,P)\in \X\times \mathcal P$$
		induces the generalized risk measure 
		$$
		\Psi(X\vert \mathcal Q)= \sup_{P\in \mathcal Q} \E^P [X] ,~~~~(X, \mathcal Q)\in \X\times 2^\P.
		$$
		For a fixed $\mathcal Q$, $\Psi(\cdot\vert \mathcal Q)$ is the robust representation of a traditional coherent risk measure of \cite{ADEH99}. This class of risk measures is the most well studied in the literature, and we will pay special attention to it in Section \ref{sec:connect}.
		
		\item 	Let $\gamma:\mathcal P\to \R$ be a non-constant function on $\mathcal P$. The  penalized-mean core
		$$\Psi(X\vert P)= \E ^P [X] - \gamma(P),~~~~(X,P)\in \X\times \mathcal P$$
		induces the generalized risk measure $$
		\Psi(X\vert \mathcal Q)= \sup_{P\in \mathcal Q} \{ \E^P [X] - \gamma(P)\},~~~~(X, \mathcal Q)\in \X\times 2^\P. 
		$$
		For a fixed $\mathcal Q$, $\Psi(\cdot\vert \mathcal Q)$ is  the robust representation of a traditional convex risk measure of \cite{FS16}. 
		
		\item  For $\alpha \in (0,1)$, the VaR core $(X,P)\mapsto \VaR_\alpha (X\vert P) $
		induces the worst-case VaR in Section \ref{sec:23}. 
%		Further, define 
%		$$
%		F(x) = \inf_{P \in \mathcal Q} F_{X\vert P}(x),~~~~x \in \R. 
%		$$
%		Under some conditions (e.g., $\mathcal Q$ is finite),  $F$ is a distribution function and we have
%		$\VaR_\alpha(F) = \sup_{P\in \mathcal Q} \VaR_\alpha (X\vert P).$ Thus,
%		\begin{equation*}
%			\Psi(X\vert \mathcal Q) = \VaR_{\alpha}\left(\inf_{P \in \mathcal Q} F_{X\vert P}\right),~~~~(X, \mathcal Q)\in \X\times 2^\P.
%		\end{equation*}
		\item For $\alpha \in (0,1)$, the ES  core  $(X,P)\mapsto \ES_\alpha (X\vert P) $ 
		induces the worst-case ES in Section \ref{sec:23}.
%		$$
%		\Psi(X\vert \mathcal Q) = \sup_{P \in \Q} \ES_\alpha(X\vert P),~~~~(X, \mathcal Q)\in \X\times 2^\P.
%		$$
%		which is the worst-case ES; see \cite{ZF09} for optimization of the worst-case ES; \cite{WZ18} contains several related examples of ES under multiple scenarios.
		%\item For $\gamma > 0$, the core (entropic type)
		%$$
		%\Psi(X\vert P)= \frac{1}{\gamma} \log \E_P[e^{-\gamma X}],~~~~(X,P)\in \X\times \mathcal P
		%$$
		%induces the generalized risk measure 
		%\begin{equation}
		%\Psi(X\vert \mathcal Q) = \frac{1}{\gamma} \sup_{P \in \Q} \log \E_P[e^{-\gamma X}],~~~~(X, \mathcal Q)\in \X\times 2^\P.
		%\end{equation}
		%\begin{equation}
		%\Psi(X\vert \mathcal Q) = \sup_{P \in \Q} \Phi(X\vert P) = \frac{1}{\gamma} \sup_{P \in \Q} \log \E_P[e^{-\gamma X}],~~~~(X, %\mathcal Q)\in \X\times 2^\P.
		%\end{equation}
	\end{enumerate}
\end{example}

\section{Three formulations of law invariance} \label{sec:law invariance}

For a given $P$, the functional $X\mapsto \Psi(X\vert P)$ is a traditional risk measure, and properties can be imposed for this traditional risk measure. 
The more interesting and non-trivial question is the interplay between $X$ and $P$ for the core $\Psi$, which we will address below. 
Since $P\in \mathcal P$ is interpreted as a scenario for us to generate a statistical model for the loss $X$, the evaluation of the risk should depend on the distribution of $X$.
Motivated by this consideration, we can consider three forms of law invariance on the generalized risk measure $\Psi$ or its core:\footnote{In this paper, all properties (Ax)  reflect how $\Psi$ reacts to $\mathcal Q$, all properties (Bx) reflect how $\Psi$ reacts to  the distributions of the risk, all properties (Cx) reflect consideration for $\Psi$ in terms of a traditional risk measure, all properties (Dx) are relevant to a mapping defined on the set of measures $\P$, and all properties (Ex) reflect consideration on decision-theoretic preference. }

\begin{enumerate}[(B1)]
	\item[(B1)] Strong law invariance:  %\com{There is no ``-" in any of these properties unless they are used as adverbial (like "a law-invariant risk measure"). I have corrected most places. Check through.}
	$\Psi(X\lvert P ) = \Psi(Y\lvert Q)$ for   $X,Y \in \X$ and 
	$P, Q \in \mathcal P$ with $X\vert _P \laweq Y  \vert  _Q$. 
	\item[(B2)] Loss law invariance:  $\Psi(X\lvert P ) = \Psi(Y\lvert P)$ for   $X,Y \in \X$ and 
	$P  \in \mathcal P$  with  $X\vert _P \laweq Y  \vert  _P$. 
	\item[(B3)] Scenario law invariance:  $\Psi(X\lvert P ) = \Psi(X\lvert Q)$ for   $X \in   \X$ and 
	$P,Q  \in \mathcal P$  with  $X\vert _P \laweq X  \vert  _Q$. 
\end{enumerate}

Clearly, (B1) is stronger than both (B2) and (B3).
Each of (B1), (B2) and (B3) reflects the consideration that the probability measures $P$ in $\Psi(X\vert P)$ is used to model the distribution of the loss $X$. 
More precisely, (B1) is an agreement of risk assessment for the same distribution across different scenarios and different losses, whereas (B2) only yields the agreement for each particular scenario,
and (B3) only yields the agreement for each particular loss. 
The following example shows that (B1), (B2) and (B3) are genuinely different concepts.  
\begin{example}\label{ex:3}
	
	\begin{enumerate}[(i)]
		\item  The cores in Example \ref{ex:worst} (i), (iii) and (iv)  
		are strongly law-invariant.
		\item The core in Example \ref{ex:worst} (ii)
		is loss law-invariant, but generally not scenario law-invariant.  
		\item Let $\beta:\mathcal X\to \R$ be a non-constant function on $\mathcal X$. The core 
		$$\Psi(X \vert P)=\E ^P [X] - \beta(X),~~~~(X,P)\in \X\times \mathcal P$$
		is scenario law-invariant, but generally not loss law-invariant.
	\end{enumerate}
\end{example}

%{\color{red}
%It is straightforward to verify that (B1) implies both (B2) and (B3). 
Since (B1) implies both  (B2) and (B3), one may wonder whether (B2) and (B3) jointly imply (B1), which turns out to be a tricky question. 
In other words, we aim to show from (B2) and (B3) that  $\Psi(X\vert P ) = \Psi (Y\vert Q)$ holds for $P,Q\in \P$ and $X,Y\in \X$ satisfying $X\vert _P \laweq Y  \vert  _Q$.
Denote by $F$ the distribution of $X$ under $P$, which is the same as that of $Y$ under $Q$. 
If there exists $Z\in \X$ which has distribution $F$ under both $P$ and $Q$, 
then we have the desired chain of equalities 
$\Psi(X\vert P) = \Psi(Z\vert P) = \Psi(Z\vert Q) = \Psi(Y\vert Q) $.
Unfortunately, the existence of such $Z$ depends on the specification of $P,Q$ and it cannot be expected in general; this problem is non-trivial and has been studied in detail by \cite{SSWW19}. 
In the result below, we show that, under the extra assumption that the measurable space $(\Omega,\mathcal F)$ is standard Borel (i.e., isomorphic to the Borel space on $[0,1]$),   
it is possible to find an intermediate measure $R$ and two random variables $Z,W\in \X$ such that the chain of equalities
$$
X\vert _P \laweq Z\vert_ P \laweq Z\vert_ R \laweq W\vert_ R \laweq W \vert_ Q \laweq Y  \vert  _Q
$$ holds, and it 
gives 
the desired statement $\Psi(X\vert P) =\Psi(Y\vert Q)$ needed for  (B1).  

%$X\vert _P \laweq Z \vert_P \laweq Z \vert _Q \laweq Y  \vert  _Q$,

\begin{theorem}\label{thm:law-inv}
	For a core $\Psi$, (B1) implies both (B2) and (B3).
	If $(\Omega,\mathcal F)$ is standard Borel, then  (B2) and (B3) together are equivalent to (B1).
\end{theorem}

\begin{remark}
	It remains unclear whether the equivalence (B2+B3)$\Leftrightarrow$(B1) holds if $(\Omega,\mathcal F)$ is not standard Borel. 
	For applications in finance and risk management, it is typically sufficient to use a standard Borel space, because one can construct countably many independent
	Brownian motions on the corresponding probability space. The assumption of a standard Borel space is used in some classic literature on risk measures, e.g., \cite{D02} and \cite{JST06}. 
\end{remark}
%}

Loss law invariance (B2) seems to be always desirable to assume in practice, because if two random losses $X$ and $Y$ share the same distribution under a chosen scenario $P$ of interest, then it is natural to assign the same risk value to these two losses.  For a fixed collection $\mathcal Q\in 2^{\mathcal P}$, this property  defines a $\Q$-based risk measure of \cite{WZ18}.
On the other hand, it may not always be desirable to assume (B3); although two scenarios may give the same distribution of a loss $X$, the riskiness may not be understood as the same, as illustrated by the following example. 

\begin{example}[(B3) is not always desirable] 
	Let $P$ represent a good economic scenario, and $Q$ represent an adverse economic scenario (e.g., COVID-19).
	Assume that the distribution of $X$ is the same under $P$ and $Q$, which means that $X$ is independent of the particular economic factor which generates $P$ and $Q$.
	The value $\Psi(X\vert P)$ quantifies the riskiness of $X$ when $P$ is the chosen scenario, and 
	$\Psi(X\vert Q)$ quantifies the riskiness of $X$ when $Q$ is the chosen scenario. 
	Since $P$ is a better economy, the risk manager may think that $X$ is more acceptable in this situation, leading to $\Psi(X\vert P)<\Psi(X\vert Q)$. 
	For instance, the core in  Example \ref{ex:worst} (ii), the robust representation of convex risk measures, reflects this consideration, and it is not scenario law-invariant.
\end{example}

Next, we collect  some  representation results based on (A1), (A3) and (B1)-(B3). %{\color{red}
	Recall that $\M$ is the set of compactly supported distributions on $\mathbb{R}$. %}.
Throughout, we define 
$$\Sigma = \{ \psi:\M\to [-\infty,\infty]\}.$$ Each mapping  $\psi \in \Sigma$ represents a traditional law-invariant risk measure  treated as a  functional on $\M$ instead of on $\X$.
\begin{proposition}\label{th:law invariance}
	Let $\Psi$ be a generalized risk measure. 
	\begin{enumerate}[(i)]
		\item $\Psi $ satisfies (A1), (A3) and (B1) if and only if there exists $\psi\in \Sigma$ such that
		$$
		\Psi(X\vert \mathcal Q) = \sup_{P\in \mathcal Q}  \psi(F_{X|P}),~~~~(X, \mathcal Q)\in \X\times 2^\P.
		$$ 
		\item   $\Psi$ satisfies (A1), (A3) and (B2) if and only if there exists   $ \{\psi_P: P\in \mathcal P \} \subseteq   \Sigma$ such that
		$$
		\Psi(X\vert \mathcal Q) = \sup_{P\in \mathcal Q} \psi_P  (F_{X|P}),~~~~(X, \mathcal Q)\in \X\times 2^\P.
		$$ 
		\item $\Psi$ satisfies (A1), (A3) and (B3) if and only if there exists  $ \{\psi_X: X\in \mathcal X \} \subseteq   \Sigma$ such that
		$$
		\Psi(X\vert \mathcal Q) = \sup_{P\in \mathcal Q}\psi_X (F_{X|P}),~~~~(X, \mathcal Q)\in \X\times 2^\P.
		$$ 
	\end{enumerate}
\end{proposition}

%
%{\color{red}\begin{proposition}\label{pro:ES axiomation}
%Let $\Psi$ be a standard generalized risk measure. $\Psi$ satisfies (A1)-(A2), (B1) if and only if its admits a representation 
%	\begin{equation}
%	\Psi(X\vert \mathcal Q) = \sup_{P\in \mathcal Q}  ES_{\alpha}  (X|P),~~~~(X, \mathcal Q)\in \X\times 2^\P
%	\end{equation}
%	for $\alpha \in [0,1].$
%\end{proposition}}

\section{Coherent  generalized risk measures}\label{sec:connect}
%\com{This section is completely rewritten.}

In this section, we pay special attention to the most important class of traditional risk measures: coherent risk measures of \cite{ADEH99}. We first provide a characterization for a generalized risk measure to have the form of coherent risk measures in Example \ref{ex:worst}, and then we discuss a few additional properties that are specific to our setting. 

\subsection{A characterization for coherent risk measures}

We give a simple characterization of the coherent risk measures in Example \ref{ex:worst} (i). 
Coherent risk measures including ES (for a fixed scenario) are the most well-studied class of risk measures in the finance and engineering literature. 
We first list some properties of traditional risk measures of \cite{ADEH99} and \cite{FS16}.   These properties are formulated for the traditional risk measure $X\mapsto \Psi(X\vert \mathcal Q)$  on $\X$ for each fixed $\mathcal Q\subseteq \mathcal P$, and we denote this by $\Psi_{\mathcal Q}$.
\begin{enumerate}[(C1)]
	\item[(C1)]  Monotonicity:  $  \Psi_{\mathcal Q}  (X) \le \Psi_{\mathcal Q}  (Y) $ for all $X,Y \in \mathcal X$ with $X \le Y$;  
	\item[(C2)]   Cash-additivity:   $\Psi_{\mathcal Q} (X+m) =  \Psi_{\mathcal Q} (X)+m$ for all $X \in \mathcal X$ and  $m \in \mathbb R$;  
	\item[(C3)]  Positive homogeneity:  $\Psi_{\mathcal Q}  (\lambda  X) = \lambda   \Psi_{\mathcal Q}  (X)$ for all $\lambda > 0$ and $X \in \mathcal X$; 
	\item[(C4)]  Subadditivity:  $\Psi_{\mathcal Q}  (X + Y ) \leq \Psi_{\mathcal Q}(X) + \Psi_{\mathcal Q}  (Y) $ for all  $X,Y \in \mathcal X$.\end{enumerate}
Following the terminology for traditional risk measures, a generalized risk measure $\Psi$ is \emph{monetary} if it satisfies (C1)-(C2) and \emph{coherent} if it satisfies (C1)-(C4).
We further state a strong property   imposed on the cores. 
\begin{enumerate}[(C1)]
	\item[(C0)] Additivity of the core: $\Psi(X+Y\vert P) = \Psi(X\vert P) + \Psi(Y\vert P)$  for all  $X,Y \in \mathcal X$ and $P\in \mathcal P$. 
\end{enumerate} 
The property (C0) will be a key property to pin down the form of coherent traditional risk measures.

\begin{theorem}\label{th:coh}
	A standard generalized risk measure $\Psi$ satisfies (A1), (A2), (B2), (C1) and (C0) if and only if it is uniquely given by \begin{align}\label{eq:coh}
		\Psi(X\vert \mathcal Q)= \sup_{P\in \mathcal Q} \E^P [X] ,~~~~(X, \mathcal Q)\in \X\times 2^\P.
	\end{align}
	Moreover, $\Psi$ in \eqref{eq:coh} satisfies (C2)-(C4).
\end{theorem}

%Since the risk measure $\Psi$ in \eqref{eq:coh} satisfies  (B1) and  (C2)-(C4), we get these properties automatically from the ones in Theorem \ref{th:coh}. 
The most important property used in Theorem \ref{th:coh} is the additivity of the core (C0), which may be seen as quite strong. As a primary example of coherent risk measures of the form \eqref{eq:coh} in financial practice,  the   Chicago Mercantile Exchange (CME) uses \eqref{eq:coh} to determine margin requirements for a portfolio of instruments; see Section 2.3 of \cite{MFE15}. In the CME approach, under each fixed scenario, the risk factors move in a particular deterministic way, and hence the portfolio loss assessment is additive; thus (C0) is natural in this context.

\subsection{Ambiguity sensitivity and comonotonic additive risk measures}
As we have seen from Example \ref{ex:3}, strong law invariance (B1) is genuinely stronger than the weaker notions of (B2) and (B3). 
In the following result, we connect weak and strong law invariance via an additional property on the generalized risk measure $\Psi$ via its core.

\begin{enumerate}[(B1)] 
	\item[(B4)] Ambiguity sensitivity: 
	For $X\in \X$, $P,Q\in \mathcal P$ and $\lambda \in [0,1]$, 
	$\Psi(X \lvert \lambda P + (1-\lambda)Q ) \ge \lambda \Psi(X \lvert P ) + (1-\lambda) \Psi(X\lvert Q)$.
	Moreover,
	$\Psi(\id_A\lvert \lambda P + (1-\lambda)Q ) = \lambda \Psi(\id_A \lvert P ) + (1-\lambda) \Psi(\id_A\lvert Q)$ for all $ A\in \mathcal F$ such that $P(A)=Q(A)$. 
\end{enumerate} 
The first statement of (B4) intuitively means that, due to ambiguity on the distribution of $X$, the risk of $X$ under a mixture is larger than the mixture of its risks under $P$ and $Q$; this is the concavity in $P$ in Proposition \ref{prop:ES}. For instance, for a random variable $X$ which is constant under both $P$ and $Q$, it may be random (Bernoulli) under $\lambda P + (1-\lambda)Q $, and hence its risk should be larger under the mixture than under the individual scenarios. 
Regarding the second statement of (B4), if the probability measures $P$ and $Q$ agree on how likely event $A$ is,
then there is no ambiguity on $A$, and its risk under a mixture should be simply a mixture of its risks under $P$ and $Q$.  
Another explanation may be illustrated by the following example. 
\begin{example}[Ambiguity sensitivity]
	Assume that $P$ is used by a risk analyst and $Q$ is used by another risk analyst.
	The manager would like to use $\lambda P + (1-\lambda)Q$, a mixture of $P$ and $Q$, to reflect the knowledge of both analysts. 
	For simplicity, the random loss $X$ is assumed to be an indicator of a loss event $A$.
	If $P$ and $Q$ give different assessments of the probability of $A$, then the manager would be worried about the discrepancy in the models, and her final risk assessment $\Psi(X\vert \lambda P + (1-\lambda)Q )$ 
	is more than $\lambda \Psi(X\vert P) + (1-\lambda)\Psi(X\vert Q)$, the weighted average of the two analysts' assessments. On the other hand, if $P$ and $Q$ give the same probability of $A$, then there is no disagreement in predicting $A$. In this case, her risk assessment of $\id_A$ is the same as the weighted average of the two analysts' assessments.
\end{example}

Another property that is essential to our next characterization result is comonotonic additivity, which is intimately linked to Choquet integrals; see e.g., \cite{WWW20}.
\begin{enumerate}[(C1)]
	\item[(C5)]  Comonotonic additivity:  $\Psi_{\mathcal Q}(X+Y) = \Psi_{\mathcal Q}(X) +  \Psi_{\mathcal Q}(Y),$ for all  $X,Y \in \mathcal X $ which are comonotonic.\footnote{Two random variables $X$ and $Y$ are \emph{comonotonic} $(X(\omega) - X(\omega '))(Y(\omega) - Y(\omega ')) \geq 0$ for all $(\omega , \omega ') \in \Omega \times \Omega$.}
\end{enumerate}

The following result characterizes loss law-invariant risk measures with ambiguity sensitivity, 
which turns out to be equivalent to strongly law-invariant risk measures without this assumption. The proof of this result is quite technical and it relies on Lyapunov's convexity theorem, as well as a few characterization results on Choquet integrals in \cite{WWW20}. %{\color{red}
	It is important to note that when we say the core satisfies some properties (C1)-(C5), it means it satisfies these properties as traditional risk measures.%}

\begin{theorem}\label{th:new2}
	For a core $\Psi$, the following are equivalent:
	\begin{enumerate}[(i)]
		\item  $\Psi$ is loss law-invariant, ambiguity sensitive, monetary, and comonotonic additive, i.e., $\Psi$ satisfies (B2), (B4), (C1), (C2) and (C5).  
		\item  $\Psi$ is strongly law-invariant, coherent, and comonotonic additive, i.e., $\Psi$ satisfies (B1) and (C1)-(C5).
		\item There exists an increasing concave function $ h  :[0,1]\to [0,1]$ with $h(0)=0=1-h(1)$ such that 
		\begin{equation}
			\label{eq:th:new2:representation}
			\Psi (X \vert P) =  \int X \d (h \circ P),~~~~(X,P)\in \X\times \P.
		\end{equation}
	\end{enumerate} 
\end{theorem}

There has been an extensive debate in both academia and industry on whether subadditivity (C4) proposed by \cite{ADEH99} is a good criterion for risk measures used in regulatory practice, as (C4) is key property which distinguishes VaR and ES; see \cite{ELW18, ESW21} and the reference therein. 
By Theorem \ref{th:new2}, from the perspective of multiple models, we can obtain (C4) by using ambiguity sensitivity (B4). 
Hence, our framework and results offer a novel decision-theoretic reason to support coherent risk measures (in particular, ES over VaR) without directly assuming subadditivity (C4).

%Some other characterization results on risk measures under strong law invariance, derived from existing ones on traditional risk measures, are collected in Appendix \ref{sec:app}, as they may be of interest in different specific applications.  

\section{Connection to decision theory}%other classic notions}
\label{sec:connection}

In this section, we discuss the connection of our generalized risk measures to classic notions in decision theory, as model uncertainty has been dealt with extensively in the decision-theoretic literature, and traditional risk measures are intimately linked to decision preferences in various forms; see e.g., \cite{DK13}.
 We first present a list of decision-theoretic criteria as examples to our framework, followed by  characterization results of two classic notions: the multi-prior expected utility of \cite{GS89} and  the variational preferences of \cite{MMR06}. 

\subsection{Examples of generalized risk measures in decision theory} 
Our framework includes many criteria in decision theory as typical examples.
Although the considerations  of these criteria are different from our paper, the following examples show the generality of our framework. 
\begin{example}\label{ex:literature}
	\begin{enumerate}[(i)]	
		\item The multi-prior expected utility of \cite{GS89} has a numerical  representation %\com{change all "in \cite{GS89}" to "of \cite{GS89}", also other places.}
		$$
		\Psi(X\vert \Q) = u^{-1}\left(\min_{P \in \mathcal Q} \E^P[u(X)]\right),%u^{-1}\(\int_{\Omega} u(X) \d P\).
		$$
		where $u$ is a strictly increasing utility function. 
		\item The variational preference of \cite{MMR06} has a numerical representation
		$$
		\Psi(X\vert \Q) = \min_{P \in \mathcal Q} \left(\E^P[u(X)] - \gamma(P)\right),%\(\int_{\Omega} u(X) \d P - \gamma(P) \).
		$$	where $u$ is a strictly increasing utility function and $\gamma: \mathcal P\to [-\infty,\infty)$ is a penalty function.  The multiplier preferences of \cite{HS01} correspond to a special choice of $\gamma$ which is the Kullback–Leibler divergence  from a reference scenario.
		\item Let $\mathcal Q\subseteq\mathcal P$ be pre-specified and $\mu$ be a probability measure on $\mathcal Q$.  The smooth ambiguity preference of \cite{KMM05} has  a numerical representation 
		$$
		\Psi(X\vert \Q) = \phi^{-1} \left(\int_{\mathcal Q} \phi \left(u^{-1}\left(\E^P[u(X)]\right)\right)\d \mu(P) \right),%\( \int_{P \in \mathcal Q} \phi \(u^{-1}\(\int_{\Omega} u(X) \d P\)\) \d \mu(P) \).
		$$ 
		where  $u$ is a strictly increasing utility function  and $\phi$ is a strictly increasing function.
		Note that in this formulation, $\mu$ needs to be  specified with $\mathcal Q$, and hence it should be considered as an input of $\Psi$ in our framework; see Section \ref{sec:conclusion} for more discussion on this.
		\item The imprecise information preference of \cite{GHTV08}  has  a numerical representation  
		$$
		\Psi(X\vert \Q) = u^{-1}\left(\min_{P \in \phi(\mathcal Q)} \E^P[u(X)]\right),
		$$
		where $u$ is a strictly increasing utility function and $\phi$ is a selecting function (assumed to exist) reflecting the decision maker's attitude to imprecision.
		\item The model misspecification preference of \cite{CHMM21}  has  a numerical representation 
		$$
		\Psi(X\vert \Q) = \min_{P \in \P} \left\{ \E^P \left[ u(X) \right] + \min_{Q \in \Q} c(P, Q) \right\},
		$$
		where $u$ is a strictly increasing utility function and $c$ is a distance on the set of measures which penalizes the model misspecification.
	\end{enumerate}	
\end{example}

\subsection{Multi-prior expected utilities}\label{sec:decision} 
\cite{GS89} proposed the notion of \textit{multi-prior expected utility} in decision theory.  Motivated by the multi-prior expected utility, we consider  a preference on $\X\times \mathcal S$ which is represented by a total pre-order  $\preceq$, where $\mathcal S$ is the collection of all finite subsets of $\P$. For tractability, we consider $\mathcal S$, all the finite subsets of $\P$, instead of $2^\P$ in this subsection. The decision is to compare a risk and a set of scenarios with another risk and another set of scenarios. This setting was studied by \cite{GHTV08}. We denote by $\simeq$ the equivalence under this preference. As above, we use $(X, P)$ if the set of scenarios has only one element $P$. For decisions among $(X_1, \Q_1), (X_2, \Q_2) \in \X\times \mathcal S$, we propose the following axioms  similar to what we have seen so far in this paper, but defined for  preferences instead of generalized risk measures. 
\begin{enumerate}
	%	\item[(E1)] We have: 
	%	\begin{enumerate}
	%		\item[(a)] either $(X_1, \Q_1) \preceq (X_2, \Q_2)$, or $(X_2, \Q_2) \preceq (X_1, \Q_1)$; 
	%		\item[(b)] if $(X_1, \Q_1) \preceq (X_2, \Q_2)$ and $(X_2, \Q_2) \preceq (X_3, \Q_3)$, then $(X_1, \Q_1) \preceq (X_3, \Q_3)$.
	%	\end{enumerate}
	%	\item[(E2)] For $\alpha \in (0,1)$, $x, y \in \R$ and fixed $\p \in \P$, if
	%	$$
	%	\alpha (X_1, \Q_1) + (1-\alpha) (x, \{\p\}) \preceq \alpha (X_2, \Q_2) + (1-\alpha) (x, \{\p\}),
	%	$$
	%	\item[(E3)] 
	
	%	\item[(E1)] Uncertainty aversion: for $X\in \X$ and  $\Q_1 \subseteq \Q_2 \subseteq \mathcal P$, we have $(X, \Q_2) \preceq (X, \Q_1) $.
	%	\item[(E2)] Uncertainty monotonicity: for $X,Y \in \X$ and $\Q \subseteq \P$, if $(X, P) \preceq (Y, P)$ for all $ P \in \Q$, then $(X, \Q) \preceq (Y, \Q)$.  
	%	\item[(E1)] Weak order: For any $P \in \P$ and any $X, Y, Z \in \X$, we have  %\com{never use bullet points in a paper.}
	%	\begin{itemize}
	%		\item either $(X, P) \preceq (Y, P)$ or $(Y, P) \preceq (X, P)$; 
	%		\item if further $(X, P) \preceq (Y, P)$ and $(Y, P) \preceq (Z, P)$, then $(X, P) \preceq (Z, P)$.
	%	\end{itemize}
	\item[(E1)] Strong law invariance: $(X, P) \simeq (Y, Q)$ for any $P, Q \in \P$ and $X, Y \in \X$ satisfying $X\vert _P \laweq Y  \vert  _Q$.
	\item[(E2)] Uncertainty aversion:   
	$(X, \mathcal Q) \preceq (X, \mathcal R) $ for any $X\in \X$ and $\mathcal R, \mathcal Q \in \mathcal S$ with $\mathcal R\subseteq \mathcal Q$.
	\item[(E3)] Uncertainty bound: for any $X \in \X$ and $\mathcal Q \in \mathcal S$, there exists some $P \in \Q$ such that $  (X, P) \preceq (X, \Q)$.  
	\item[(E4)] Independence: for any $P,Q \in \P$, any $X, Y \in \X$ satisfying  $X\vert _Q \laweq Y  \vert  _Q$, and any $\alpha \in (0,1)$ we have $(X, P) \preceq (Y, P) \Longleftrightarrow (X, \alpha P+(1-\alpha)Q) \preceq (Y, \alpha P+(1-\alpha)Q)$.
	\item[(E5)] Continuity: for any $P,Q,R \in \P$ and any $X \in \X$, if $(X, P) \preceq (X, Q)  \preceq (X, R)$, then there exists $\alpha \in [0,1]$ such that $ (X, \alpha P+(1-\alpha)R) \simeq (X, Q)$. %\com{I reformulated and removed some axioms. Check for correctness.}
	%	\item[(E6)] Non-degeneracy: for any $P \in \P$, there exist some $X, Y \in \X$ such that $(X, P) \preceq (Y, P)$.
	
\end{enumerate}

%Using Theorem \ref{th:worstcase}, we directly have:

Proposition \ref{th:preference} illustrates a decision-theoretic characterization for the multi-prior expected utility. The proof is based on Theorem \ref{th:worstcase} and the classic result of \cite{NM44}.  
\begin{proposition}
	\label{th:preference}
	A preference $\preceq$ on $\X\times \mathcal S$ satisfies (E1)-(E5) if and only if
	it is a multi-prior expected utility, i.e., there exists a function $u: \R \to \R$ such that 
	\begin{equation}\label{eq:decision}
		(X_1, \Q_1) \preceq (X_2, \Q_2)~ \Longleftrightarrow ~\min_{P\in \mathcal Q_1}  \E^P[ u (X_1)] \leq \min_{P\in \mathcal Q_2} \E^P[ u (X_2)]. 
	\end{equation} 
\end{proposition}

The strong law invariance (E1), which allows us to translate $\preceq$ to a preference on the set of  distributions on the real line,  is crucial to the representation result.
(E2) and (E3) are reasonable for uncertainty-averse decision makers, and they correspond to (A1) and (A3), respectively, in the framework of generalized risk measures. (E4) and (E5)  correspond to the independence axiom and the continuity axiom of \cite{NM44}, respectively.

%{\color{red}
%\begin{theorem}
%	A preference $\preceq$ on $\X\times 2^\P$ satisfies (E1)-(E3) if and only if it is a variational preference, i.e. there exists a utility functional $u: \M \to \R$ and a penalty function $\gamma: \P \to \R$ such that
%	\begin{equation}
%	(X_1, \Q_1) \preceq (X_2, \Q_2) \Longleftrightarrow \inf_{P\in \mathcal Q_1} \{ u\(F_{X_1|P}\) + \gamma(P)\} \leq \inf_{P\in \mathcal Q_2} \{ u\(F_{X_2|P}\) + \gamma(P)\}.
%	\end{equation} 
%\end{theorem}
%}
%\com{What examples should we take? We should naturally connect to \cite{MMR06}. Add something here for sure.... } 
%\com{It seems difficult to characterize \cite{MMR06} because (B5) in Proposition \ref{th:robust} is hard to translate.} 
%\end{comment}

\subsection{Robust generalized risk measures} 
\label{sec:robust} 

In addition to the worst-case generalized risk measure characterized in 
Theorem \ref{th:worstcase}, another popular form of risk measures involving multiple probability measures arises from the robust representation of convex risk measures as in Example \ref{ex:worst} (ii).
More precisely, a traditional convex risk measure $\rho$ of \cite{FS16} takes the form, for some $\mathcal Q \subseteq \P$,
\begin{equation}
	\label{eq:robust-rep}
	\rho(X) = \sup_{P\in \mathcal Q} \{ \E^P[X] - \gamma(P)\},~~~~ X   \in \X,
\end{equation}
where    $\gamma:\P\to (-\infty,\infty]$ is a penalty function. 
Moreover, the variational preference of \cite{MMR06} takes a similar form to \eqref{eq:robust-rep} with the mean $\E$ replaced by an expected utility;\footnote{In the setting of numerical representation of preferences, a negative sign needs to be applied to the generalized risk measures to transform it to a preference functional.} see Example \ref{ex:literature} (ii).
Inspired by \eqref{eq:robust-rep} and the variational preference of \cite{MMR06}, we consider generalized risk measures with the form, 	for some $\psi\in \Sigma$,
\begin{equation} \label{eq:robust-rep2}
	\Psi(X\vert \mathcal Q) = \sup_{P\in \mathcal Q} \{ \psi  (F_{X|P}) - \gamma(P)\},~~~~(X, \mathcal Q)\in \X\times 2^\P.
\end{equation}
Clearly, if $\psi$ is the mean functional, then \eqref{eq:robust-rep2} yields the traditional (convex) risk measure \eqref{eq:robust-rep} for a given $\mathcal Q$. 
The generalized risk measure in \eqref{eq:robust-rep2} is loss law-invariant (B2) but neither scenario law-invariant (B3) nor strongly law-invariant (B1). 
In order to characterize \eqref{eq:robust-rep2}, we further impose the following technical property, which says that the difference  between the values of the core evaluated on $P$ and $Q$ for identically distributed losses only depends on $P$ and $Q$ but not the random losses. 

\begin{enumerate} 
	
	\item[(B5)]  If $X\vert _P \laweq Y  \vert  _Q$
	and $Z\vert _P \laweq W  \vert  _Q$, then $\Psi(X\lvert P )- \Psi(Y\lvert Q) = \Psi(Z\lvert P )- \Psi(W\lvert Q).$
\end{enumerate} 
%But this property looks so bad and economically not interpretable.

\begin{proposition}\label{th:robust}
	Let $\Psi$ be a generalized risk measure. $\Psi$ satisfies (A1), (A3), (B2) and (B5) if and only if there exist a penalty function $\gamma: \P \rightarrow \R$ %satisfying $\gamma(P) = \infty$ for $P \notin \Q$ 
	and some $\psi \in \Sigma$ such that the representation \eqref{eq:robust-rep2} holds.  %\com{I still feel it is incorrect or meaningless, because standardization may force $\psi$ and $\gamma$ to be both $0$. If you want to include this part in the paper, you need to make it correct.}
	%	\begin{equation}
	%	\Psi(X\vert \mathcal Q) = \sup_{P\in \mathcal Q} \{ \psi  (F_{X|P}) - \gamma(P)\},~~~~(X, \mathcal Q)\in \X\times 2^\P.
	%	\end{equation}
\end{proposition}

%As (B5) is harder to interpret and justify, the characterization with (B5) implies that the robust generalized risk measure is not a good criterion as the worst-case one.
Property 
(B5) can be roughly interpreted as that the magnitude of penalization for a given scenario $P$ is independent of the risky position $X$ being evaluated. This property may be seen as a bit artificial. 
Our characterization in Proposition \ref{th:robust} 
is mainly motivated by the great popularity of the robust representation of convex risk measures and variational preferences, and we omit a detailed discussion of the economic  desirability or undesirability of (B5).

\section{Concluding remarks}\label{sec:conclusion}

The new framework of generalized risk measures introduced in this paper allows for a unified systemic formulation of measures of risk and uncertainty. 
Our results are only first attempts to understand the new setting, and many further questions arise, especially regarding the interplay between the risk variable $X$ and the uncertainty collection $\mathcal Q$ for a generalized risk measure. Both new economic and mathematical questions arise as the new functionals are more sophisticated than traditionally studied objects by definition. 

Worst-case generalized risk measures are characterized with a few axioms in Theorem \ref{th:worstcase}. 
Another popular way of handling model uncertainty is to use a weighted average of the risk evaluations. In the case of a finite collection $\mathcal Q$, we can always use an arithmetic average as the risk evaluation, that is, to generate $\Psi$ via its core by
$$
\Psi (X\vert \mathcal Q) = \frac 1{|\mathcal Q|}\sum_{Q\in \mathcal Q} \Psi(X\vert Q) .
$$
Certainly, such a formulation does not satisfy (A1) but it satisfies (A2) and (A3).
In general, to allow for different weights and infinite collections, 
one needs to associate each collection $\mathcal Q$ with a measure as in  \eqref{eq:avgES} or the smooth ambiguity preference of \cite{KMM05} in Example \ref{ex:literature} (iii).
Such a measure can either be pre-specified or treated as an input of $\Psi$, thus slightly extending our framework.

We studied several most popular properties such as law-invariance, coherence and comonotonic additivity, but many more properties on the new framework remain to be explored, as the literature on traditional risk measures is very rich. In particular, the desirability of   theoretic properties in risk management practice requires thorough study, as they may have  different interpretations from its traditional counterpart. For instance, additivity of the core may be sensible in our framework (Theorem \ref{th:coh}) and it nicely connects to the scenario-based margin calculation used by CME. However, such a property is not desirable for traditional risk measures, as it essentially forces the risk measure to collapse to the mean; see e.g., \cite{LM21} and  \cite{CGLL21}.
 
 Finally, we mention that in some formulations of generalized risk measures, not all choices of the input scenario $\mathcal Q$ are economically meaningful.
 In particular, for a given penalty function $\gamma$ on $\mathcal P$, 
 the core $ \Psi(X\vert P)= \E ^P [X] - \gamma(P)$ in Example \ref{ex:worst} (ii) or $ \Psi(X\vert P)= \E ^P [u(X)] - \gamma(P)$ in Example \ref{ex:literature} (ii) is not meant to be used directly with a single $P$; the use of $\gamma$ already implicitly implies that there is some level of model uncertainty, and it is supposed to be coupled with the worst-case operation. 
 The value $\Psi(X\vert P)$ for a standalone $P$ is thus difficult to interpret, and should not be used for decision making without properly specifying the uncertainty collection $\mathcal Q$.
On the other hand, such a situation does not happen in, for instance, the worst-case  or   average-type generalized risk measures based on traditional risk measures, such as the worst-case ES.

\subsection*{Acknowledgments}
We thank the review team, Camilo Garcia Trillos, Fabio Maccheroni, Andrea Marcina and Thorsten Schmidt for comments which led to various improvements of the paper. YL gratefully acknowledges financial support from the research startup fund at The Chinese University of Hong Kong, Shenzhen and the Natural Sciences and Engineering Research Council of Canada (RGPIN-2017-04054). 
%, and the Air Force Office of Scientific Research under award number FA9550-20-1-0397. Additional support is gratefully acknowledged from NSF 1915967 and 2118199. 
RW acknowledges financial support from the Natural Sciences and Engineering Research Council of Canada (RGPIN-2018-03823, RGPAS-2018-522590).

\appendix
\section{Proofs of all technical results}

\begin{proof}[Proof of Proposition \ref{prop:ES}]
For a fixed $P$, convexity of  $\ES_\alpha(\cdot\vert P)$ is well-known since $\ES_\alpha(\cdot\vert P)$   is a coherent risk measure; see e.g., Theorem 4.52 of \cite{FS16}.
The non-convexity and non-concavity of $\VaR_\alpha(\cdot\vert P)$ are due to the fact that $\VaR_\alpha(\cdot\vert P)$ has a non-convex and non-concave distortion function; see e.g., Theorem 3 of \cite{WWW20}.
For a fixed $X$, note that a mixture on scenarios leads to a mixture of the distribution of $X$, that is, the distribution of $X$ under $\lambda P + (1-\lambda)Q$ is $\lambda F_{X|P} +(1-\lambda)F_{X|Q}$.
Thus, concavity with respect to scenarios corresponds to mixture concavity studied in \cite{WWW20}.  
Again, using Theorem 3 of \cite{WWW20}, $\ES_\alpha(X\vert \cdot)$ is concave and $\VaR_\alpha(X\vert \cdot)$ is neither convex nor concave.
\end{proof}

\begin{proof}[Proof of Theorem \ref{th:worstcase}]
	\begin{enumerate}[(i)] \item The ``if" statement can be directly checked since \eqref{eq:def_unc} satisfies (A1)-(A2) for a standard generalized risk measure. 
	We show the ``only if" statement below. 
	Using (A1), we have $\Psi(X\vert \Q) \ge \Psi(X\vert P)$ for all $P \in \Q$, which implies ``$\ge$" in (\ref{eq:def_unc}). Define $s_X= \sup_{P \in \Q} \Psi(X\vert P)$. As $s_X$ is a constant random variable under every $P \in \Q$, we have 
	$\Psi(X\vert P) \leq s_X= \Psi(s_X \vert P) $ for all $P \in \Q$.
	Using (A2), we have
	$
	\Psi(X\vert \Q) \leq \Psi(s_X \vert \Q) = s_X.
	$
	Thus, ``$\le$" in \eqref{eq:def_unc} follows. 
	\item For the ``only if" statement, (A1) gives the ``$\ge$" direction of \eqref{eq:def_unc} and 
	the (A3) gives the ``$\le$" direction of \eqref{eq:def_unc}. The ``if" statement is straightforward to check. \qedhere\end{enumerate}
\end{proof}

\begin{proof}[Proof of Theorem \ref{thm:law-inv}]
	The first statement can be directly checked. For the second statement, it suffices to show 
	(B2+B3)$\Rightarrow$(B1).
	Take $P,Q\in \P$ and $X,Y\in \X$ such that $X\vert _P \laweq Y  \vert  _Q$, and we denote this common distribution by $F$.
	As explained above, we aim to show that $\Psi(X\vert P ) = \Psi (Y\vert Q)$. 
	
	%Suppose that $(\Omega,\mathcal F)$ is a Radon space.
	Let $P'=(P+Q)/2$ which is a probability measure dominating both $P$ and $Q$.  
	Since both $(\Omega,\mathcal F,P)$ and $(\Omega,\mathcal F,Q)$ are atomless, 
	so is $(\Omega,\mathcal F,P')$. 
	Hence, there exist  iid uniform $[0,1]$ random variables $U$ and $V$ under $P'$. 
	Take an arbitrary $x\in (0,1)$ and  define a probability measure $R$  as the regular conditional probability	$$R(A) = P'(A\vert U=x),~~~A\in \mathcal F.$$
	Note that $R$ is a well-defined  probability measure since the measurable space $(\Omega,\mathcal F)$ is standard Borel (see e.g., Theorem 5.1.9   of \cite{D10}).
	It is clear that $R$ is atomless since $V$ is uniformly distributed under $R$. Moreover, $R$ and $P'$ are mutually singular.  
	Since $P,Q\ll P'$, we know that $R$ and $P$ are mutually singular, and so are $R$ and $Q$. 
	By Remark 3.13 and Theorem 3.17 of \cite{SSWW19}, 
	there exists a random variable $Z$ such that the distribution of $Z$ is $F$ under both $P$ and $R$.
	Similarly, there exists a random variable $W$ such that the distribution of $W$ is $F$ under both $Q$ and $R$. 
	Therefore, we obtain the chain of equalities
	$$
	X\vert _P \laweq Z\vert_ P \laweq Z\vert_{R} \laweq W\vert_{R} \laweq W \vert_ Q \laweq Y \vert_Q,
	$$
	which implies 
	$$
	\Psi(X\vert P) =\Psi(Z\vert P) =\Psi(Z\vert {R}) =\Psi(W \vert {R}) =\Psi(W\vert Q) =\Psi(Y\vert Q).
	$$
	Hence, $\Psi$ satisfies (B1).	 
\end{proof}

\begin{proof}[Proof of Proposition \ref{th:law invariance}] 
	For the proof of (i), an application of Theorem \ref{th:worstcase} shows that $\Psi $ satisfies (A1) and (A3) if and only if 
	\begin{equation}\label{eq:th:law invariance1}
		\Psi(X\vert \mathcal Q) = \sup_{P\in \mathcal Q} \Psi  (X\vert P),~~~~(X, \mathcal Q)\in \X\times 2^\P.
	\end{equation}
	It remains to show that $\Psi$ satisfies (B1) if and only if there exists a function $\psi \in \Sigma$ such that $\Psi (X\vert P) = \psi(F_{X|P})$ for any $(X, P)\in \X\times \P$. The ``if" statement can be checked directly. 
	
	Now we prove the ``only if" statement. For any $F \in \M$, there exists a random variable $X \in \X$ and a probability measure $P \in \P$ such that $F$ is the distribution of $X$ under $P$. According to (B1), $\Psi(X \vert P)$ is irrelevant to the selection of $X$ and $P$. 
	Hence, we define a functional $\psi \in \Sigma$ via
	\begin{equation}
		\label{eq:th:law invariance2}
		\psi(F) = \Psi(X \vert P), ~~~~ F \in \M.
	\end{equation}
	Combining \eqref{eq:th:law invariance1} and \eqref{eq:th:law invariance2}, we complete the proof of (i). One can prove (ii) and (iii) similarly. 
\end{proof}

\begin{proof}[Proof of Theorem \ref{th:coh}]
	It is straightforward to check that the risk measure $\Psi$ defined by \eqref{eq:coh} satisfies (A1), (A2), (B1), (C1)-(C4) and (C0).
	Next, we show that these properties pin down \eqref{eq:coh}. 
	Using Theorem \ref{th:worstcase}, it suffices to show that, with  (B2), (C1) and (C0), the mapping $X\mapsto \Psi(X\vert P)$ has to be the expectation under $P$.  It is a well-known result that a monotone, standard, law-invariant, and additive functional has to be the mean; see e.g., the proof of Lemma A.1 of \cite{WZ21}.  Hence, these properties are enough to pin down \eqref{eq:coh}.
\end{proof}
\begin{proof}[Proof of Theorem \ref{th:new2}]
	First, we show the equivalence (ii)$\Leftrightarrow$(iii). 
	To show the direction (ii)$\Rightarrow$(iii), based on (B2) and (C1)-(C5), for a fixed $P\in \P$, we get, by Theorem 1, Lemma  2 and Theorem 3 of \cite{WWW20}, 
	there exists an increasing concave function $h_P:[0,1]\to [0,1]$ with $h_P(0)=0=1-h_P(1)$ such that  
	\begin{equation}\label{eq:th1-1}
		\Phi(X\vert P) = \int X \d (h_P\circ P) ,~~~~(X,P)\in \X\times \P. 
	\end{equation}
	For every $x\in [0,1]$ and $P,Q\in \P$, 
	there exist  $A,B\in \mathcal F$ with  $P(A)= Q(B)=x$ by Theorem 1 of \cite{D12}.
	Strong law invariance  of $\Phi$ implies 
	$$ \Phi(\id_A \vert  P) = \Phi(\id_B \vert  Q) = h_P(P(A)) = h_Q(Q(B)).$$ 
	Thus, $h_P(x) = h_Q(x)$ for all $x \in [0,1]$.
	This shows that $h_P$ does not depend on $P$, and writing $h=h_P$,
	\eqref{eq:th1-1} leads to \eqref{eq:th:new2:representation}.
	
	To show the direction (iii)$\Rightarrow$(ii), for fixed $P\in \P$, using Theorem 3 of \cite{WWW20}, $\Phi(\cdot\vert P)$ is subadditive. Monotonicity, cash-additivity, positive homogeneity and comonotonic additivity follow from Theorems 4.88 and 4.94 of \cite{FS16}. Therefore, $\Phi$ is coherent and comonotonic-additive. Strong law invariance follows from the fact that  $X\vert P\laweq Y\vert Q$ implies
	\begin{align*}
		\int  X \d (h\circ P) &= \int_{-\infty}^{0} ( h \circ P (X \geq x) - 1 ) \d x +  \int_{0}^{\infty} h \circ P (X \geq x)  \d x  
		\\ &= \int_{-\infty}^{0} ( h \circ Q (Y \geq x) - 1 ) \d x +  \int_{0}^{\infty} h \circ Q (Y \geq x)  \d x  
		= \int Y \d (h \circ Q).
	\end{align*}
	
	Next, we show (i)$\Leftrightarrow$(iii). To show (iii)$\Rightarrow$(i), as the other properties are straightforward to check, we only show ambiguity sensitivity. Because of the cash-additivity of \eqref{eq:th:new2:representation}, it suffices to check it for $X\ge 0$. For $X\ge 0$ and $P,Q\in \P$, 
	\begin{align*}
		\Psi(X\vert  \lambda P + (1-\lambda) Q)  & =  \int  X \d (h\circ (\lambda P + (1-\lambda) Q) )\\   &=    \int_{0}^{\infty} h    (\lambda P  (X \geq x)  + (1-\lambda) Q (X \geq x)  )  \d x   
		\\&  \ge  \int_{0}^{\infty}  ( \lambda (h    \circ  P)  (X \geq x)   + (1-\lambda)( h \circ Q) (X \geq x)   )    \d x  
		\\&= \lambda  \int  X \d( h \circ P )+ (1-\lambda ) \int  X \d( h \circ Q )
		\\ & = \lambda \Psi(X\vert  P) +  (1-\lambda)  \Psi (X\vert   Q).
	\end{align*}
	Moreover, for all $ A\in \mathcal F$ and $P \in \P$, we compute
	\begin{align*}
		\Psi(\id_A\lvert P) = \int_{0}^{\infty} h(P(\id_A \geq x))  \d x   
		= \int_{0}^{1} h(P(\id_A \geq x))  \d x  
		= \int_{0}^{1} h(P(A))  \d x  = h(P(A)).
	\end{align*} 
	For all $ A\in \mathcal F$ and $P,Q \in \P$ such that $P(A)=Q(A)$, we have
	\begin{align*}
		\Psi(\id_A\lvert \lambda P + (1-\lambda)Q) &=  h\left((\lambda P + (1-\lambda)Q)(A)\right)= h(P(A))\\
		&= \lambda h(P(A)) + (1-\lambda) h(Q(A))= \lambda \Psi(\id_A \lvert P ) + (1-\lambda) \Psi(\id_A\lvert Q).
	\end{align*} 
	
	To show (i)$\Rightarrow$(iii), by Theorem 1 and Lemma 2 of \cite{WWW20}, based on (B2), (C1), (C2) and (C5), for a fixed $P\in \P$, there exists an increasing function $h_P:[0,1]\to [0,1]$ with $h_P(0) = 0 = 1-h_P(1)$ such that 
	\begin{equation}\label{eq:th:new2:prime}
		\Psi(X\vert P) = \int X \d (h_P\circ P),~~~~(X,P)\in \X\times \P. 
	\end{equation}
	Also note that $h_P(x)= \Psi(\id_A\vert P)$ for all $P\in \P$ and $A\in \mathcal F$ with $P(A)=x$. Further, note that for any $P \in \P$ and any $t \in [0,1]$,
	$$
	\Psi(\id_{\{U \leq t \} } \vert P) = \int \id_{\{U \leq t \} } \d (h_P \circ P) = h_P(t).
	$$
	%Hence, (C7) implies that for any $P \in \P$ and $U$ is a uniform distribution on $[0,1]$ under $P$,
	%$$
	%\lim_{t\rightarrow 1} \Psi(\id_{\{U \leq t \} } \vert P) = \lim_{t\rightarrow 1} h_P(t) = h_P(1-) = 1 = h_P(1).
	%$$
	
	Lyapunov's convexity theorem (Theorem 5.5 of \cite{R91}) 
	states that the set $R(P,Q):=\{ (P(A),Q(A)): A\in \mathcal F\}$ is closed and convex. Hence, as $(0,0), (1,1) \in R(P, Q)$, for every $x\in [0,1]$ and $P,Q\in \P$, 
	there exists $A\in \mathcal F$ with  $P(A)= Q(A)=x$.
	Hence, the second condition in (B4) implies   
	\begin{equation}
		\label{eq:certainty2}
		h_{\lambda P + (1-\lambda )Q} (x) = \lambda h_{P} (x) + (1-\lambda)h_Q(x).
	\end{equation}
	
	Assume $P\ne Q$. There exists $B \in \mathcal F$ such that $P(B)>Q(B)$, which also implies $P(B^c)<Q(B^c)$.
	Therefore, $R(P,Q)$ contains at least one point above the diagonal line 
	and at least one point below the diagonal line. 
	Using  Lyapunov's convexity theorem again, since  $$(0,0), (1,1), (P(B), Q(B)), (P(B^c), Q(B^c)) \in R(P, Q),$$ for any $\epsilon \in (0, \frac{1}{2})$, there exists a ``rectangle" set in $R(P, Q)$, i.e., there exists some $\delta'>0$ such that 
	\begin{equation}\label{eq:rectangle}
		\left\{(x+\delta, x-\delta): x \in (\epsilon, 1-\epsilon), \delta \in (-\delta', \delta')\right\} \subseteq R(P,Q). 
	\end{equation}
	Hence, for any $\delta \in (0, \delta')$, for any $ x \in (\epsilon, 1-\epsilon)$, there exists some $B_\delta \in \F$ satisfying $P(B_\delta)=x+\delta$ and $Q(B_\delta)=x-\delta$.      
	The first condition in (B4) further gives
	\begin{align*}
		h_{  P/2  +  Q/2} (x ) &= h_{P/2  +  Q/2} \left(\frac {P(B_\delta)} 2 + \frac {Q(B_\delta)} 2 \right)  \\&\ge  
		\frac 12 h_{  P } (P(B_\delta))+
		\frac 12 h_{  Q } (Q(B_\delta))
		=
		\frac 12 h_{  P } (x+\delta )+
		\frac 12 h_{  Q } (x-\delta ).
	\end{align*}
	It then follows from \eqref{eq:certainty2} that
	$$ 
	h_Q(x) - h_{Q} (x-\delta) \geq h_{P} (x+\delta) - h_{P} (x),
	$$
	which implies 
	\begin{equation}\label{eq:th:new2:concave}
		\frac 1 \delta ( h_Q(x) -  h_{  Q } (x-\delta ))\ge 
		\frac 1 \delta (  h_{  P } (x+\delta )- h_{P} (x)).
	\end{equation}
	By the symmetry of $P$ and $Q$, for any $x \in (\epsilon, 1-\epsilon)$ and $\delta \in (0, \delta')$ sufficiently small, we similarly have 
	\begin{equation}\label{eq:th:new2:concave:sym}
		\frac 1 \delta ( h_P(x) -  h_P (x-\delta ))\ge 
		\frac 1 \delta (  h_Q (x+\delta )- h_Q (x)).
	\end{equation}
	For any $x \in (\epsilon, 1-\epsilon)$ and $\delta \in (0, \delta')$ sufficiently small, subscribing $x+\delta$ in \eqref{eq:th:new2:concave:sym}, we have
	\begin{equation}\label{eq:th:new2:concave2}
		\frac 1 \delta ( h_P(x+\delta) -  h_P (x))\ge 
		\frac 1 \delta (  h_Q (x+2\delta )- h_Q (x+\delta)),
	\end{equation}
	Combining \eqref{eq:th:new2:concave} and \eqref{eq:th:new2:concave2}, for any $x \in (\epsilon, 1-\epsilon)$ and $\delta \in (0, \delta')$ sufficiently small, we have
	\begin{equation}
		\frac 1 \delta ( h_Q(x) -  h_{  Q } (x-\delta ))\ge 
		\frac 1 \delta (  h_Q (x+2\delta )- h_Q (x+\delta)),
	\end{equation}
	which implies that $h_Q$ is concave on $(\epsilon, 1-\epsilon)$. As $\epsilon \in (0,\frac{1}{2})$ is arbitrary, we have $h_Q$ is concave on $(0,1)$. Concavity implies that $h_Q$ is absolutely continuous on $(0,1)$. Similarly, $h_P$ is also concave and hence absolutely continuous on $(0,1)$. 
	
	Note that $h_P$ and $h_Q$, as increasing functions, have derivatives almost everywhere. Letting $\delta \downarrow 0$ in \eqref{eq:th:new2:concave} gives 
	$$
	\frac{\d }{\d x} h_Q(x) \ge \frac{\d }{\d x}h_P(x)
	\mbox{~~~~ for a.e.~$x\in (0,1)$}.
	$$
	%where $\frac{\d}{\d x}$ is the differential operator. 
	By symmetry in the positions of $P$ and $Q$, we have 
	$$
	\frac{\d}{\d x} h_P(x) \ge \frac{\d}{\d x}h_Q(x) \mbox{~~~~ for a.e.~$x\in (0,1)$},
	$$
	which further implies that 
	\begin{equation}\label{eq:th:new2:derivative}
		\frac{\d }{\d x} h_P(x) = \frac{\d}{\d x}h_Q(x) \mbox{~~~~ for a.e.~$x\in (0,1)$}.
	\end{equation}
	Using the Newton-Leibnitz formula, for any $x \in (0,1]$, we have 
	\begin{equation}\label{eq:NewtonL}
		h_P(1-) - h_P(x) = \int_{x}^{1} \frac{\d }{\d t} h_P(t) \d t = \int_{x}^{1} \frac{\d }{\d t} h_Q(t) \d t = h_Q(1-) - h_Q(x).
	\end{equation}

	We proceed to prove that $h_P(1) = h_P(1-)$ for any $P \in \P$. It is clear that $h_P(1) \geq h_P(1-)$ because $h_P$ is increasing. Write
	$$
	\tilde{h}_P(x) = \left\{
	\begin{aligned} 
		& h_P(x), ~ x \in [0, 1);\\
		& h_P(1-), ~ x = 1.
	\end{aligned}
	\right.
	$$
	Hence, for any $X \in \X$ satisfying $0 \leq X \leq 1$, we have
	\begin{equation}\label{eq:tilde_h}
		\begin{aligned}
			\Psi(X \vert P) &= \int X \d h_P \circ P\\
			&= \int X \d (h_P(1) - \tilde{h}_P(1)) \circ P + \int X \d \tilde{h}_P \circ P\\
			&= (h_P(1) - \tilde{h}_P(1)) \cdot \essinf (X \vert P) + \int X \d \tilde{h}_P \circ P\\
			&= (h_P(1) - \tilde{h}_P(1)) \cdot \essinf (X \vert P) + \int_{0}^{1} \tilde{h}_P \circ P(X \geq x) \d x\\
			& \leq (h_P(1) - \tilde{h}_P(1)) \cdot \essinf (X \vert P) + \tilde{h}_P(1). 
		\end{aligned}
	\end{equation}

	For any $\lambda \in [0,1]$, let $X$ be a random variable satisfying $X \vert P \sim \text{Bernoulli}(\lambda)$. On one hand, according to \eqref{eq:tilde_h}, we have
	\begin{equation}\label{eq:h_cont_at_one_LHS}
		\begin{aligned}
			\Psi(X \vert P) & \leq \tilde{h}_P(1). 
		\end{aligned}
	\end{equation}

	On the other hand, we define two probability measures: $Q(\cdot) = P(\cdot \vert X = 0)$ and $R(\cdot) = P(\cdot \vert X = 1)$. In fact, we have $P = (1-\lambda) Q + \lambda R$. Hence, $Q$ and $R$ are mutual singular with $Q(X = 0) = 1$ and $R(X = 1) = 1$. Hence, $\Psi(X \vert Q) = 0$ and $\Psi(X \vert R) = 1$. According to (B4), we have
	\begin{equation}\label{eq:h_cont_at_one_RHS}
		\Psi(X \vert P) \geq (1-\lambda) \Psi(X \vert Q) + \lambda \Psi(X \vert R) = \lambda.
	\end{equation}
	Combining \eqref{eq:h_cont_at_one_LHS} and \eqref{eq:h_cont_at_one_RHS}, we have $\tilde{h}_P(1) \geq \lambda$ for any $\lambda \in [0,1]$, which implies that $h_P(1-) = \tilde{h}_P(1) = 1 = h_P(1)$. 
	
	Finally, combining with \eqref{eq:NewtonL}, we have $h_P = h_Q$ on $(0,1]$. Hence, together with $h_P(0)=h_Q(0) = 0$, we have $h_P=h_Q$.  
	%\com{GAP: as we do not know whether $h_P$ is absolutely continuous (e.g., Cantor function is not a.c.), we cannot apply the Newton-Leibniz formula to obtain $h_P = h_Q$. Nevertheless, one should expect the equality in view of \eqref{eq:th:new2:concave}. Technical (may be tedious) calculus might be helpful.} 
	This shows that $h_P$ does not depend on $P$, and we write $h=h_P$. As $h$ is concave on $(0,1)$, $h(1-)=h(1)$ and $h(0+)\geq h(0)$, we know that $h$ is concave on $[0,1]$. Hence, we complete the proof of  \eqref{eq:th:new2:representation}.	
\end{proof}
\begin{proof}[Proof of Proposition \ref{th:preference}]
	It is straightforward to check that the multi-prior expected utility in \eqref{eq:decision} satisfies (E1)-(E5). 
	Below we will show  the representation \eqref{eq:decision} from (E1)-(E5).
	
	By (E1), to compare $(X,P)$ with $(Y,Q)$, it suffices to compare the distributions $F_{X|P}$ and $F_{Y|Q}$ as elements of $\mathcal M$,  the set of distributions on $\R$. Hence, the restriction of $\preceq$  to $\{(X,P):X\in \X,~P\in \mathcal P\}$ is described equivalently by  a binary relation $\preceq^*$ on $ \mathcal M$, via 
	$$
	(X, P) \preceq (Y,Q) ~\Longleftrightarrow~  F_{X|P} \preceq^* F_{Y|Q}.
	$$
	By letting $F_{X|P}=F$, $F_{Y|P}=G$ and $F_{X|Q}=F_{Y|Q}=H$,  we can translate (E4) into the following property: For any $F,G,H\in \mathcal M$ and 
	$\alpha \in (0,1)$, it holds $$F \preceq^* G \Longleftrightarrow   \alpha F+(1-\alpha)H \preceq^* \alpha G+(1-\alpha) H;$$ this is the independence axiom of \cite{NM44} on $\preceq^*$. Similarly, (E5) can be translated into the continuity axiom of \cite{NM44} on $\preceq^*$.
	Using the Von Neumann-Morgenstern utility theorem, there exists a function $u:\R\to \R$ such that
	\begin{equation}\label{eq:compareutility}
		(X, P) \preceq (Y, Q) ~\Longleftrightarrow~  F_{X|P} \preceq^* F_{Y|Q}~\Longleftrightarrow~  \E^P[u(X)] \le \E^Q[u(Y)].
	\end{equation}
	Next, we consider two general objects $(X_1,\mathcal Q_1)\in \X\times \mathcal S$ and $(X_2,\mathcal Q_2)\in \X\times \mathcal S$. First, by (E2) and (E3), there exist $Q^*_1\in \mathcal Q_1$ and $Q^*_2\in \mathcal Q_2$ such that 
	$$
	(X_1,Q^*_1) \simeq (X_1,\mathcal Q_1) \mbox{~~and~~}(X_2,Q^*_2) \simeq (X_2,\mathcal Q_2)
	$$
	Moreover, by (E2), we have $ (X_1,Q^*_1) \simeq  (X_1,\mathcal Q_1) \preceq (X_1,Q) $ for all $Q\in \mathcal Q_1$.
	By using \eqref{eq:compareutility},  we have
	$$
	\E^{Q^*_1} [u(X_1)] \le \E^{Q} [u(X_1)]~~~\mbox{for all $Q\in \mathcal Q_1$.}
	$$
	Thus, $\E^{Q^*_1} [u(X_1)]= \min_{Q\in \mathcal Q_1} \E^Q[u(X_1)]$.
	Similarly,   $\E^{Q^*_2} [u(X_1)]= \min_{Q\in \mathcal Q_2} \E^Q[u(X_2)]$.
	Suppose that $\E^{Q^*_1}[u(X_1)] \le \E^{Q^*_2}[u(X_2)]$. By using \eqref{eq:compareutility} again, we have 
	$$
	(X_1,Q^*_1) \preceq  (X_2,Q)~~~\mbox{for all $Q\in \mathcal Q_2$.}
	$$
	Using (E3),  this implies $(X_1,\mathcal Q_1) \simeq(X_1,Q^*_1)  \preceq  (X_2,\mathcal Q_2)  $. 
	Reverting the positions of $(X_1,\mathcal Q_1)$ and $(X_2,\mathcal Q_2)$, we obtain that  
	\eqref{eq:decision}  holds true. 
\end{proof}

\begin{proof}[Proof of Proposition \ref{th:robust}]
	Proposition \ref{th:law invariance} shows that $\Psi$ satisfies (A1), (A3) and (B2) if and only if there exists  $ \{\psi_P: P\in \mathcal P \} \subseteq  \Sigma$ such that
	$$
	\Psi(X\vert \mathcal Q) = \sup_{P\in \mathcal Q} \Psi(X \vert P),~~~~(X, \mathcal Q)\in \X\times 2^\P
	$$ 
	and
	\begin{equation}\label{eq:th:robust1}
		\Psi(X \vert P) = \psi_P  (F_{X|P}),~~~~(X, P)\in \X\times \P.
	\end{equation}
	It remains to show that $\Psi$ satisfies (B5) if and only if there exists $\gamma: \P \to (-\infty, \infty]$ and $\psi \in \Sigma$ such that 
	\begin{equation}\label{eq:th:robust2}
		\Psi(X \vert P) = \psi_P  (F_{X|P}) = \psi(F_{X|P}) - \gamma(P),~~~~(X, P)\in \X \times \P.
	\end{equation} 
	The ``if" statement can be checked directly. Now we proceed to prove the ``only if" statement.

	Assume that (B5) holds. For any $X, Y, Z, W \in \X$ and $P, Q \in \P$ satisfying $X\vert _P \laweq Y  \vert  _Q$ and $Z\vert _P \laweq W  \vert  _Q$, according to \eqref{eq:th:robust1} we have
	$
	\psi_P  (F_{X|P}) - \psi_Q  (F_{Y|Q}) = \psi_P  (F_{Z|P}) - \psi_Q  (F_{W|Q}).
	$ 
	Write $F_1 = F_{X \vert P} = F_{Y \vert Q}$ and $F_2 = F_{Z \vert P} = F_{W \vert Q}$. Hence, for any $F_1, F_2 \in \M$ and $P, Q \in \P$, we have
	$$
	\psi_P  (F_1) - \psi_Q  (F_1) = \psi_P  (F_2) - \psi_Q  (F_2),
	$$
	which means $\psi_P(F) - \psi_Q(F)$ is a constant for any $F \in \M$. That is, there exists a functional $g: \P \times \P \to \R$ such that
	$$
	\psi_P(F) - \psi_Q(F) = g(P, Q), ~~~~ P, Q \in \P, ~~ F \in \M.
	$$
	Fix $P \in \P$ and set $Q = \delta_0$, the Dirac measure at zero. Then
	$$
	\psi_P(F) = \psi_{\delta_0}(F) + g(P, \delta_0),~~~~ F \in \M.
	$$
	Define $\psi(\cdot) = \psi_{\delta_0}(\cdot)$ and $\gamma(\cdot) = -g(\cdot, \delta_0)$. We have
	\begin{equation}\label{eq:th:robust3}
		\psi_P(F) = \psi(F) - \gamma(P),~~~~ F \in \M.
	\end{equation}
	Since \eqref{eq:th:robust3} holds for any $P \in \P$, 
	$$
	\psi_P(F) = \psi(F) - \gamma(P),~~~~ (P, F) \in \P \times \M.
	$$
	For any $X \in \X$ and $P \in \P$, we consider $F = F_{X|P}$ and hence have
	$$
	\Psi(X \vert P) = \psi_P(F_{X \vert P}) = \psi(F_{X \vert P}) - \gamma(P),~~~~ (X, P) \in \X \times \P,
	$$
	which is exactly \eqref{eq:th:robust2}.
\end{proof}

\end{document}